%% file: main.tex
\documentclass[11pt]{article}
\usepackage{dblfloatfix}
\usepackage{blkarray}
\usepackage{amsmath}
\usepackage{todonotes}
\usepackage{xfrac} 
\usepackage[noend]{algpseudocode}
\usepackage{algorithm}
\usepackage{complexity}
\usepackage{comment}
\usepackage{graphics}
\usepackage{framed}
\usepackage{xspace}
\usepackage{subcaption}
\usepackage[english]{babel}

\usepackage[a4paper,top=2cm,bottom=2cm,left=3cm,right=3cm,marginparwidth=1.75cm]{geometry}

\usepackage[utf8]{inputenc}
\usepackage{mathrsfs}
\usepackage{dsfont}
\usepackage{amsfonts}
\usepackage{complexity}
\usepackage{latexsym}
\usepackage{amssymb}    
\usepackage{amsmath}
\usepackage{graphicx}
\usepackage[colorlinks=true, allcolors=blue]{hyperref}
\input{letterfonts}

\input{macros}
\usepackage{cleveref}
\usepackage{babel}
\usepackage{amsthm}
\usepackage{thmtools,thm-restate}
\usepackage{lineno}
\makeatletter
\newcommand{\customlabel}[2]{%
   \protected@write \@auxout {}{\string \newlabel {#1}{{#2}{\thepage}{#2}{#1}{}} }%
   \hypertarget{#1}{\textbf{\normalsize{#2}}}
}
\makeatother
\newtheorem{definition}{Definition}[section]
\newtheorem{theorem}{Theorem}[section]
\newtheorem{corollary}{Corollary}[theorem]
\newtheorem{lemma}[theorem]{Lemma}

\newtheorem{remark}[theorem]{Remark}
\date{}

\title{Fast exact algorithms via the Matrix Tree Theorem}
\author{V. Arvind, Srijan Chakraborty, Samir Datta and Asif Khan\\
	\texttt{arvind@imsc.res.in, \{srijanc,sdatta,asifkhan\}@cmi.ac.in}}

\begin{document}
	\maketitle
	
	\abstract 
	\input{abs}

	\input{intro}
	\input{org}
	\input{prelims}
	\input{matrixTree}
	\input{maxmatch}

	\input{thmproofs}

	\input{dirHamPcount}
	\input{concl}

	\bibliographystyle{alpha}
	\bibliography{biblio}
	\appendix
	\input{mttU}
\end{document}

%% file: letterfonts.tex
	\newcommand{\bbN}{\mathbb{N}}

\newcommand{\mcO}{\mathcal{O}}



%% file: macros.tex
\newcommand{\OO}{\mathcal{O}}

\renewcommand{\epsilon}{\varepsilon}

\renewcommand{\epsilon}{\varepsilon}
\newcommand{\ignore}[1]{}

%% file: abs.tex
Fast exact algorithms are known for Hamiltonian paths in undirected
and directed bipartite graphs through elegant though involved algorithms that
are quite different from each other. We devise algorithms that are simple
and similar to each other while having the same upper bounds. The common
features of these algorithms is the use of the Matrix-Tree theorem and sieving
using roots of unity.

Next, we use the framework to provide alternative algorithms to 
count perfect matchings in bipartite graphs on $n$ vertices, i.e., computing 
the $\{0,1\}$-permanent of a square $n/2 \times n/2$ matrix which runs 
in a time similar to Ryser. 

We demonstrate the flexibility of our method by counting the number of ways 
to vertex partition the graph into $k$-stars (a $k$-star consist of a tree 
with a root having $k-1$ children that are all leaves). Interestingly, 
our running time improves to $\OO^*((1+\epsilon_k)^n)$ with $\epsilon_k \rightarrow 0$
as $k \rightarrow \infty$. 

As an aside, making use of Bj\"orklund's algorithm for exact counting perfect 
matchings in general graphs we show that the count of maximum matchings 
can be computed in time $\OO^*(2^\nu)$ where $\nu$ is the size of a maximum 
matching. The crucial ingredient here is the famous Gallai-Edmonds 
decomposition theorem. 

All our algorithms run in polynomial space.


%% file: intro.tex
\section{Introduction}
The determinant is ubiquitous in Mathematics and 
Theoretical Computer Science. It is the signed sum of the product of the 
entries of an $n \times n$ matrix,
 picked by a permutation by ranging over all the permutations.
The corresponding unsigned sum is called the permanent. 

The determinant is computable in polynomial time \cite{vzGG}
 and even efficiently parallelizable (see e.g.  \cite{Berk,MV}).
On the contrary, due to Valiant \cite{Valiant} we know that 
the permanent is $\shP$-hard to compute, i.e., for any $\NP$-language $L$
 counting the number of witnesses for a word being in $L$ is reducible to 
computing the permanent. 
The contrast between
the determinant and the permanent plays an important role in Algebraic Complexity
Theory. 

In more graph theoretic terms, the permanent of a $\{0,1\}$ matrix $A$ is
the count of the perfect matchings in the bipartite graph with $A$ as its 
biadjacency matrix\footnote{The biadjacency matrix has rows enumerated by
one part of the bipartite graph and columns enumerated by the other part with
all entries zero except the ones which correspond to an edge between
two vertices of opposite parts.}.

Another pertinent object is the so called Hamiltonian which differs from the
permanent (or the determinant) where we restrict the permutations to 
$n$-cycles. Naturally, the Hamiltonian of the adjacency matrix of a graph
is the count of all Hamiltonian cycles (i.e., simple cycles that include all 
vertices) in the graph. Like the permanent, the Hamiltonian is $\shP$-hard,
but detecting the existence of a Hamiltonian path is an $\NP$-hard problem,
unlike the detection of a perfect matching in bipartite graphs, which has been
known to be in polynomial time since the advent of the Hungarian Method \cite{kuhn}.

We will use a seemingly innocuous result to provide simple and uniform proofs
of the best known upper bounds
on these, two, hard to compute objects -- the permanent and the Hamiltonian
along with some of their restrictions and variants.
The result we use is the so-called Matrix Tree Theorem which counts
the number of spanning trees in a graph by reducing it to a single determinant
computation. Notice that here we are ambiguous about the exact version of 
matrix tree theorem we use -- apart from the theorem for undirected
graphs there is also one for counting directed spanning trees (or arborescences)
in directed graphs -- and we shall make use of both of these.

The trick is to view the patterns we want to count viz. perfect matchings/
Hamiltonian paths as a bijective subgraph of certain spanning trees in a 
graph closely related to the input graph. But then we'll have to ``sieve-out''
 all but these certain trees so that the count of remaining trees equals the
count of the patterns.

\subparagraph{Our technique and a summary of main results}
We associate a suitable weight with each edge
of the graph in terms of certain variables. Recall that for weighted graphs
with adjacency matrix $A$, the weight of a
permutation $\pi$ is the product of the entries $A_{i,\pi(i)}$ over $i \in [n]$.
Thus, if we plug in a product of certain variables as the weight of an edge, the
weight of each permutation will turn out to be a monomial. The matrix tree
theorem ensures that each (uncancelled) monomial occurring in a certain determinant corresponds to a tree. Thus, if we sum up
over the monomials, assuming that the signs
of all monomials agree, then the determinant will yield a polynomial, in which
none of the monomials cancel out. Furthermore, in each monomial the sequence 
of exponents of various variables will closely depend on the shape of the tree.
Thus, for example, if we can ensure that (a) the polynomial is homogeneous
(i.e., each monomial degree is the same) and (b) the monomials of interest
correspond to a pre-specified degree sequence (such as, all of them are $2$)
and (c) the unwanted monomials have at least one variable with an odd
exponent, then using well-known properties of monomials we can extract
their coefficients. This approach turns out to be elegant, trading off
some tricky algorithmic ideas (such as \cite{Williams09,FGS}) for the
simplicity of plugging in $\pm 1$ for each of the variables. This 
substitution ensures that all monomials which have degree $2$ in each
variable yield $2^{\text{\#variables}}$ while other monomials cancel out.
This will yield a scaled version of the desired monomial count.

We distil this argument into a recipe (which we describe and illustrate in
\hyperref{app:undirHamP}{Appl 0.1} via the example of counting Hamiltonian paths
in undirected graphs). Essentially, the best algorithm known for this
problem is via Held-Karp \cite{HeldKarp} that runs in $\OO^*(2^n)$
time. This is the ``best'' known algorithm because there is no
algorithm that runs faster than this by a factor of
$\Omega(2^{\epsilon n})$ for any constant $\epsilon > 0$. Our algorithm also runs in $\OO^*(2^n)$ time but is fundamentally
different from Held-Karp.

Next, we provide two more examples of this recipe application that match
the best known running times of $\mcO^*(2^{n/2})$ and
$\mcO^*(3^{n/2})$ for finding a Hamiltonian path in bipartite 
undirected and bipartite directed graphs, respectively. While we cannot
count Hamiltonian paths in these two cases because we are unable to ensure that
the signs of the monomials are the same, using the 
\emph{isolation} lemma \cite{MVV} we can extract a witnessing Hamiltonian
path in each of these cases, thereby solving the decision problem. We also present
some generalisations of these results.

Next, we turn to perfect matchings and show that perfect matchings can be 
counted as fast as Ryser's algorithm in the bipartite case through our recipe. We also show that matchings of size $k$ for any $k$ can be counted in Ryser time in bipartite graphs, i.e., $\OO(2^{n/2})$ time. It is known that this is $\#\Wone$-$\hard$ \cite{Curt}.

We also show that maximum matchings in bipartite graphs can be counted
in time $\OO(2^\nu)$ where $\nu$ is the size of the maximum matching.
Bj\"{o}rklund \cite{Bjorklund14} had already shown that ``Ryser-time'' 
suffices for counting perfect
matchings in general, i.e., non-bipartite graphs as well. As a short aside we show that
the Gallai-Edmonds decomposition -- a venerable technique in matching theory --
allows us to lift the maximum matching result to general graphs 
using Bj\"{o}rklund's result as a black box.

In perfect matchings we cover all the vertices with $n/2$ disjoint edges. A perfect matching is also called a perfect edge cover. 
As our final result we extend this result to count the number perfect $k$-star covers of graph on $n$ vertices, where a $k$-star is a tree with one vertex of degree $k-1$ attached to 
$k-1$ leaves\footnote{We do not claim this as an induced structure, i.e., the 
star could induce edges other than the $k-1$-spokes}. Somewhat paradoxically
the running time of our algorithm is $(1+\epsilon_k)^n$ for some $\epsilon_k \rightarrow 0$ as $k \rightarrow \infty$. Note that detecting whether a graph $G$ is $k$-star coverable is $\NP$-$\hard$ \cite[Problem GT12]{GJ} for $k\ge 3$ unlike detecting a perfect matching which is polynomial time solvable. We remark that all our algorithms run in polynomial space.

We summarise our results in \cref{tab:results}.
\input{table}
\paragraph{Related Work}
Perhaps the first example of an algebraic algorithm for a graph problem dates back to Kirchoff's matrix tree theorem from 1847. It essentially reduces the problem of counting the number of spanning trees in an undirected graph to computing the determinant of a matrix associated with graph. More recently, in the context of parametrized and exact algorithms, algebraic methods have enjoyed tremendous success, starting with Koutis' reduction for the problem of detecting a $k$-path in a graph to detecting whether there exists a multilinear monomial in a polynomial that itself is obtained from a matrix multiplication operation~\cite{Koutis08}. With further refinement from Williams, one can detect $k$-paths in general graphs in $\OO^*(2^k)$ time with high success probability~\cite{Williams09}. Here, the polynomial (called walk polynomial) enumerates all $k$-length walks, and multilinear monomials in that correspond precisely to $k$-paths. The  underlying algebraic algorithm is that, given a polynomial over a field of characteristic $2$, (the polynomial is implicitly given by a monotone algebraic circuit) whether there exists a multilinear monomial of degree $k$ in the monomial-sum expansion of the polynomial in time $\OO^*(2^k)$ randomized time using group algebras. For polynomials given by arithmetic circuits over fields of arbitrary characteristic (where field operations can be done efficiently), Brand el al.~\cite{BDH18}, and Arvind et al.~\cite{ACDM18} gave randomized $\OO^*(4.32^k)$ time algorithm for detecting degree $k$-multilinear monomials in the polynomial. See also~\cite{Pratt19,Brand22, BKS23} for more on $k$-multilinear monomial detection.

In general, the enumerating polynomial method involves designing a polynomial that enumerates general enough substructures of the input graph, of a specified size as monomials, such that, for the desired substructures the corresponding monomials have a specific shape, e.g., multilinear, while undesired substructures do not have the specific shape. For another example, Bj\"{o}rklund's $\OO^*(1.66^n)$ time algorithm for finding Hamiltonian cycles~\cite{Bjorklund14}, that broke the long standing barrier of $O^*(2^n)$ running time, uses a polynomial that enumerates all labelled cycle covers in a graph. Since Hamiltonian cycles are also cycle covers, they too appear in the polynomial. Crucially, only for the non-Hamiltonian cycle covers the monomials corresponding to them pair up and thus cancel when evaluated in a characteristic $2$ field. And the algebraic techniques in algorithms has found many fruitful applications since, e.g., exact algorithms for problems like Steiner tree, disjoint paths, independent set, chromatic number, graph motif problem, etc. See~\cite{KoutisWilliams} for a survey.


The determinantal sieving technique \cite{EKW24}, introduced very recently, encompasses together many of the algebraic algorithms in a common framework. 
Given a polynomial $P(X)$ of degree $d$ over a field of characteristic $2$, and a linear matroid $M=(X,\mathcal{I})$ on $X$ of rank $k$, we can test whether there exists a multilinear monomial $m$ of degree $k$ in $P(X)$ such that the set of variables that appear in $m$ are independent in $M$ in $\OO^*(d2^k)$ time.
The idea is to design a linear matroid $M$ over $X$ and $P(X)$ will be a generating polynomial enumerating candidate subgraphs, that can be computed efficiently (e.g., determinant of a matrix) such that the monomials corresponding to the correct subgraphs are multilinear as well as the variable appearing in them are independent in $M$. In particular, \cite{EKW24}, show fast exact algorithms for finding long paths in undirected graphs using this technique. Their algorithm relies on constructing an intricate matroid $M$ of the determinantal sieving. The running time comes from an involved analysis, and gives the best running time. Though, we focus only on Hamiltonian paths and that too in bipartite (di)graph, our expositions are much simpler, yet we get the same best running time. 

Bj\"{o}rklund et al. \cite{BKK17}, gave the best running time exact algorithm for finding Hamiltonian cycles in directed bipartite graphs. They too use the matrix-tree theorem, for modular counting of Hamiltonian paths, and counting $k$-out-branchings. But, for detecting the Hamiltonian cycles, they modify the Laplacian matrix, so that over exponentially many determinant sums, only the monomials that correspond to Hamiltonian cycles (a single cycle that is a cycle cover), survive over a field of characteristic $2$ of suitable size. Their algorithm is inspired by a specific interpretation of the directed matrix-tree theorem that directly analysing the determinant of the reduced Laplacian. Their modification to the usual Laplacian, which they call quasi-Laplacian do the job, but become very complicated. Compared to their construction, ours is much simpler, and the proofs are straightforward, using the arguably simpler Cauchy-Binet theorem based proof of directed matrix-theorem. We look for spanning trees in the squared graph, and rely on the shape of the monomials corresponding to the Hamiltonian paths, and sieve them by summing over all substitutions of cube roots of unity.

\subparagraph{Matrix tree theorem in parametrized and exact algorithms} We recount that matrix tree theorem that has been used already in the context of parametrised and exact algorithms algorithms--most notably for designing fixed parameter tractable algorithms for connectivity problems, such as counting Hamiltonian paths, Steiner tree, connected feedback vertex sets, etc., by doing dynamic programming over tree decompositions. Bodlaender el al. employ the matrix tree theorem to keep track of connected objects that form the desired subgraphs, in the dynamic programming over tree decompositions~\cite{BCKN15}. Wlodarczyk further improved the algorithms using Clifford algebras~\cite{Wlodarczyk16}.

\subparagraph{Vertex variables vs edge variables} A notable difference between the existing algebraic algorithms for graph problems, e.g., detecting Hamiltonian paths or long paths (as well as $k$-paths)~\cite{BKK17,EKW24}, and our algorithms, is that we use vertex variables instead of edge variables. For example, in~\cite{EKW24} (determinantal sieving), for detecting long $st$-paths, edge variables are used. In~\cite{Bjorklund14,BKK17} too, edge variables are used, and the monomials in the generating polynomials are by design multilinear. They rely on cancellations in characteristic $2$ field, amongst monomials that corresponds to non-desirable subgraphs, and use PIT-lemma to check if in the final polynomial any good monomials remain. For sieving good monomials, we rely on homogeneity of all the monomials and good monomials being multi-$k$-adic (multi-quadratic, multi-cubic for $k=2,3$ respectively). 

Our construction don't necessarily fit in the determinantal sieving paradigm. The desirable monomials in our polynomials are not multilinear and it is not clear what would be the construction of the matroid $M$ over the variables. 

Also, applying the cone-size based monomial detection technique for detecting desirable monomials~\cite{FGS}, would give us a worse-off running time that we get by sieving via roots of unity.



%% file: table.tex
\begin{table}[h!]
  \centering
  \caption{Our Results}
  \label{tab:results}
  \begin{tabular}{ |l|ccc| } 
    \hline
    \textbf{Problem} & \textbf{Bound} & \textbf{Previous} & \textbf{Our Construction} \\
    \hline
    Undir. Ham. Path Count & $2^{n}$ & \cite{HeldKarp,Barvinok} & \hyperref[app:undirHamP]{Appl. 0.1} \\
    Dir. Ham. Path Count & $2^{n}$ & \cite{HeldKarp,Barvinok} & \hyperref[app:dirHamP]{Appl. 0.2} \\
    \hline
    Undir. Bip. Ham. Path Dec. & $2^{n/2}$ & \cite{Bjorklund14} & \hyperref[app:undirBipHamP]{Appl. 1.1} \\
    Dir. Bip. Ham. Path Dec.& $3^{n/2}$ & \cite{BKK17} & \hyperref[app:dirBipHamP]{Appl. 1.2} \\
    \hline
    Bip. Perfect Matching Count & $2^{n/2}$ & \cite{Ryser} & \hyperref[app:BipPMCount]{Appl. 2.1}  \\
    Maximum Matching Count & $2^{\nu}$ & & \hyperref[app:MMCount]{Appl. 2.2}  \\
    Perfect $k$-star Covering Count & $2^{n/(\frac{k}{3\log{k}})}$ & & \hyperref[app:kStarCount]{Appl. 2.3}  \\
    \hline
  \end{tabular}
\end{table}

%% file: org.tex
\subsection{Organization}
We provide a recipe for our approach by counting Hamiltonian paths in undirected graphs in~\cref{sec:matrixTree}, followed by illustrating the method for detecting Hamiltonian paths in bipartite graphs (in directed and undirected) in~\hyperref[app:undirBipHamP]{Appl. 1.1}, ~\hyperref[app:dirBipHamP]{Appl. 1.2}. We then generalize these to graphs with a given large independent set in the graph, in~\hyperref[app:undirIndHamP]{Appl. 1.1.1},~\hyperref[app:dirIndHamP]{Appl. 1.2.1}. In the next section,~\cref{sec:match}, we move onto counting matchings and covering problems. We start with counting perfect matchings in bipartite graphs, in~\hyperref[app:BipPMCount]{Appl. 2.1}, followed by maximum matching count in~\hyperref[app:MMCount]{Appl. 2.2}, and perfect $k$-star covering count in~\hyperref[app:kStarCount]{Appl. 2.3}. We also show how to count the number of $k$-matchings in bipartite graphs in~\hyperref[app:BipkMCount]{Appl. 2.4}. 

For completeness, we provide the preliminaries in~\cref{sec:prelims}. We provide formal theorems and their in~\cref{sec:thmproofs}. Finally, we conclude with~\cref{sec:concl}.        

%% file: prelims.tex
\section{Preliminaries}\label{sec:prelims}
\paragraph{Cauchy-Binet Theorem}
\begin{theorem}[Cauchy-Binet Theorem~\cite{Stanley}]
	Let $A=(a_{ij})$ be an $m\times n$ matrix, with
	$1\leq i\leq m$ and $1\leq j\leq n$. Let $B=(b_{ij})$ be an $n\times m$ matrix with $1\leq i \leq n$ and
	$1\leq j\leq m$. (Thus $AB$ is an $m\times m$ matrix.) If $m > n$, then $\det(AB)=0$. If $m\leq n$,
	then
	\[
		\det(AB)=
		\sum\limits_S(\det A[S])(\det B[S]),
	\]
	where $S$ ranges over all $m$-element subsets of $\{1,2,\ldots,n\}$
\end{theorem}

\paragraph{Roots of unity}
For any $k$, there are $k$ complex roots of unity, viz., $1,\omega,\omega^2,\ldots,\omega^{k-1}$, that satisfy, $1+\omega+\omega^2+\ldots+\omega^{k-1} = 0$ where $w$ is a primitive $k$th root of unity. 
For any square matrix $M$ of size $n$ whose entries are from $\mathbb{Z}[\omega]$, $\det(M)$ can be computed in polynomial time. We can treat $\omega$ as a formal variable, and then compute the determinant of the matrix $M$ treating its entries as polynomials over $\omega$, in polynomial time using~\cite{AAM}. Finally, dividing the resulting determinant with $\omega^k-1$, that can be done in polynomial time using~\cite{HAB}, gives us the desired determinant of $M$ in $\mathbb{Z}[\omega]$.       
\paragraph{Matrix Tree Theorems}
For a graph $G$, any acyclic and connected subgraph of $G$ is called a \emph{spanning tree} of a $G$. For a directed graph $G=(V,E)$ and specific vertex $r$ of $G$, an \emph{in-arborescence} rooted at $r$, is a spanning tree of the underlying undirected graph, such that in the in-degree of each vertex in the tree is exactly one, except for $r$.
   
We now define incidence and in-incidence matrices of undirected and directed graphs respectively.
\begin{definition}[Incidence matrix]
	For any undirected graph $G=(V,E)$ with $n$ vertices and $m$ edges, its weighted incidence matrix $B(G)$ is a matrix of size $n\times m$ under an orientation $\mathfrak{o}$ of its edges is,
	\[
	B(G)[v,e] = 
	\left\{
	\begin{array}{ll}
		w(e) & \mbox{ if } e \mbox{ is an \emph{outgoing} edge from } v \\
		-w(e) & \mbox{ if } e \mbox{ is an \emph{incoming} edge at } v \\
		0 & \mbox{ otherwise} 
	\end{array}
	\right.
	\]
	where $w(e)$ is the weight of the edge $e$.
\end{definition}

\begin{definition}[In-incidence matrix]
	For any directed graph $G=(V,E)$, its in-incidence matrix, $N_{in}(G)$ is a rectangular matrix of size $n\times m$, with rows and columns indexed by vertices and edges of $G$. It is defined as follows.
	\[
	N_{{in}}(G)[v,e] = 
	\left\{
	\begin{array}{ll}
		1 & \mbox{ if } e \mbox{ is an \emph{incoming} edge at } v \\
		0 & \mbox{ otherwise} 
	\end{array}
	\right.
	\]
\end{definition}
For any undirected graph $G$, the \emph{spanning tree polynomial} of $G$, $\mbox{sp}(G)$ over $X=\{x_e|e\in E(G)\}$ is
\[
	\mbox{sp}(G) = \sum\limits_{T\in\mathcal{T}(G)}\prod\limits_{e\in T} x_e,
\]
where $\mathcal{T}(G)$ is the set of all spanning trees of $G$.
\begin{restatable}{theorem}{mttU}[Matrix Tree Theorem]\label{lem:mttU}
	Let $G$ be an undirected graph. Let $X=\{x_e|e\in E(G)\}$ be a set of formal variables. Let $B(G)$ be the weighted incidence matrix of $G$ under an arbitrary orientation $\mathfrak{o}$, such that any edge $e$ is given the weight $x_e$. Let $A_r$ be the matrix obtained from $B(G)$ after removing the row corresponding to a vertex $r\in V(G)$. Let $C$ be the unweighted (or equivalently, when all edge weights are $1$) incidence matrix of $G$ under the orientation $\mathfrak{o}$ (after the row corresponding to $r$ is removed). Then, 
	\[
		\mbox{sp}(G) = \det(A_rC^\top).
	\]    
\end{restatable}
For any directed graph $G=(V,E)$ and a vertex $r$ of $G$, the $r$-\emph{in-arborescence polynomial} of $G$, over $X=\{x_e|e\in E(G)\}$ is,
\[
	\mbox{Arb}_r(G) = \sum\limits_{S\in\mathcal{T}_r(G)}\prod\limits_{e\in S}x_e,
\]
where $\mathcal{T}_r(G)$ is the set of all $r$-rooted in-arborescences in $G$.

\begin{restatable}{theorem}{mttD}\label{thm:mttD}[Directed Matrix Tree Theorem]
	Let $G=(V,E)$ be a directed graph. Let $X=\{x_e|e\in E(G)\}$. Let $B(G)$ be the weighted incidence matrix of the undirected graph underlying $G$, under the same orientation as the direction of edges in $G$, such that the weight of any edge $e$ is $x_e$. Let $A_r$ be the matrix obtained from $B(G)$ by removing the row corresponding to $r$. Let $C$ be the in-incidence matrix of $G$, after the row corresponding to $r$ has been removed. Then,
	\[
		\mbox{Arb}_r(G) = \det(A_rC^\top).
	\]
\end{restatable}
The proofs of both the directed and undirected matrix tree theorems are provided in \cref{app:a} for completeness.
\paragraph{Isolation lemma}
\begin{lemma}[Isolation lemma~\cite{MVV}]
	\label{lem:isolate}
	Let $S$ be a finite set, $S=\{z_1,z_2,\ldots,z_n\}$, and let $\mathcal{F}$ be a family of subsets of $S$, i.e., $\mathcal{F}\subseteq2^S$. When the elements of $S$ are assigned integer weights uniformly and independently at random from $[4n]$,
	\[
		\Pr[\mbox{There is a unique minimum weight set in }F]\geq \frac{3}{4},
	\] 
	where the weight of a set is the sum of the weights of the elements in the set.
\end{lemma}
\input{gallai}
\paragraph{Notations}
We denote $[n]$ to be the set $\{1,2,\ldots,n\}$. Given a polynomial $p(X)$ over $X$ variables, and a monomial $m$ over $X$, $[m](p(X))$ denotes the coefficient of the monomial $m$ in $p(X)$.

%% file: gallai.tex
\paragraph*{Matchings}
Let $G=(V,E)$ be an undirected graph on $n$ vertices. A matching $M$ is a subset of edges $E$ of $G$ such that no two edges in $M$ share a common endpoint. A matching $M$ of $G$ is a maximum matching of $G$ if for all matchings $M'$ of $G$, we have $|M|\ge|M'|$. Moreover, $M$ is said to be a perfect matching of $G$ if $|M|=n/2$. Edmonds alogirithm \cite{Edmonds} is a polynomial time algorithm to find a maximum matching.. Denote by $\nu=\nu(G)$ to be the size of the maximum matching in $G$. 
\paragraph*{Gallai-Edmonds Decomposition \cite{LP}}
Let $G=(V,E)$ be an undirected graph. Then, $V(G)$ can be partitioned into three sets $A(G),C(G),D(G)$ as follows: $D$ is the set of vertices that are left unmatched in at least one maximum matching of $G$, $A=\left(V\setminus D\right)\cap N_G(D)$ where $N_G(D)$ is the set of neighbors of $D$, and $C=V\setminus\left(D\cup C \right)$. Then the following holds:
\begin{enumerate}
	\item\label{item:1} The components of $D(G)$ are factor-critical graphs, i.e., each component has an odd number of vertices, and when any one of these vertices is removed, there is a perfect matching of the remaining vertices. 
	\item\label{item:2} The subgraph induced by $C(G)$ has a perfect matching.
	\item\label{item:4} Every maximum matching in $G$ has the following structure: it consists of a near-perfect matching of each component of $D(G)$, a perfect matching of $C(G)$, and edges from all vertices in $A(G)$ to distinct components of $D(G)$.
	\item\label{item:5} If $D(G)$ has $k$ components, then we have $\nu(G)=\frac{1}{2}(|V(G)|-k+|A(G)|)$. In particular, if the components of $D$ are $\{D_i\}_{i\in[k]}$, then $\nu=\frac C 2+A+\sum_i\frac{|D_i-1|}2$.
\end{enumerate}
We know that the Gallai-Edmonds Decomposition of a graph $G$ can be computed in $\poly(n)$ time using Edmonds' algorithm \cite{Edmonds}.

%% file: matrixTree.tex
\section{Recipe and applications}\label{sec:matrixTree}

\noindent\customlabel{app:undirHamP}{Application 0.1: Counting Hamiltonian Paths in undirected graphs}

Given an undirected unweighted graph $G=(V,E)$ and two of its vertices $s$ and $t$, we are interested in counting the number of $st$-Hamiltonian paths in $G$.
Since, Hamiltonian paths are also spanning trees, the set of all spanning trees of $G$ contains all the Hamiltonian paths in $G$, as well.   
Let $A_s$ and $C$ be the weighted (weight of edge $e$ is $x_e$) and unweighted incidence matrices of $G$ respectively (with row corresponding to $s$ removed). Let's call the weighted version of $G$ as $H$. 

From matrix-tree theorem, we know that $\det(A_sC^\top)$ gives the spanning tree polynomial of $G$ in edge variables . That is,
\begin{align}
	\label{eq:spHam}
	\det(A_sC^\top) =  \sum_{T\in\mathcal{T}(G)}\prod_{e\in T} x_e
\end{align}    
where $\mathcal{T}(G)$ is the set of all  spanning trees of $G$. 

Substituting $x_ux_v$ in place of $x_e$ for each edge $e=\{u,v\}$ in~\cref{eq:spHam}, we can write
\begin{align}
	\label{eq:spvHam}
	\det(A_sC^\top) = \sum_{T\in\mathcal{T}(G)}\prod_{v\in V}x^{\deg_T{(v)}}_v.
\end{align}
Here, $\deg_T(v)$ is the degree of the vertex $v$ in $T$. 

Notice that for any $st$-Hamiltonian path $T$, its contribution in the sum in the RHS of the above equation is exactly $x_sx_t\prod_{v\in V\setminus\{s,t\}}x^2_v$. This is because the degree of each internal vertex on a Hamiltonian path is exactly $2$, and the endpoints of the path, $s$ and $t$ have degree $1$ each. 
Since for every $st$-Hamiltonian path, we have a contribution of $x_sx_t\prod\limits_{v\in V\setminus\{s,t\}}x^2_v$, the coefficient of $x_sx_t\prod\limits_{v\in V\setminus\{s,t\}}x^2_v$ in the polynomial $\sum_{T\in\mathcal{T}(G)}\prod\limits_{v\in V}x^{\deg_T{(v)}}_v$ gives us the number of $st$-Hamiltonian paths in $G$. 

Multiplying $x_sx_t$ on both sides of the~\cref{eq:spvHam} has the effect that the monomials corresponding to Hamiltonian paths become homogeneous, i.e., $\prod_{v\in V}x^2_v$.

Thus, we have that
\begin{align}
	\label{eq:ham}
	\#\mbox{HamPaths}(G,s,t) = [\Pi_{v\in V}x^2_v]\left(x_sx_t\det(A_sC^\top)\right).
\end{align}

In the following, we will see a way to compute the coefficient of $\Pi_{v\in V}x^2_v$ in $x_sx_t\det(A_sC^\top)$.

We claim that
\begin{align}
	\label{eq:mnml}
	\sum_{\bar{x}\in\{-1,1\}^n}x_sx_{t}\det\left(A_sC^\top\right) = 2^n\cdot [\Pi_{v\in V}x^2_v]\left(x_sx_t\det(A_sC^\top)\right), 
\end{align}
where $\bar{x}$ is the $n$-tuple consisting of variables for each vertex of $G$ (arranged in lexicographic order).

To see this, note that any monomial in $x_sx_t\det(A_sC^\top)$ that corresponds to non-Hamiltonian spanning tree has at least one leaf vertex, say $u$, that is distinct from both $s$ and $t$, and thus in the monomial degree of $x_u$ is one. Clearly, such monomials will get cancelled in the summation over all substitution of vertex variables from $\{-1,1\}$. On the other hand, monomials that have all variables appearing with even degree will contribute $1$ for each substitution. Moreover, the only monomials with degree of all variables even in $x_sx_t\det(A_sC^\top)$, are precisely those that correspond to $st$-Hamiltonian paths, viz. $\Pi_{v\in V}x^2_v$. Thus, summing over all $2^n$ substitutions of vertex variables, we obtain~\cref{eq:mnml}.

Finally, from~\cref{eq:ham,eq:mnml}, we have
\begin{align}
	\label{eq:uhamcount}
	\#\mathrm{HamPaths}(G,s,t)=\frac{1}{2^n}\sum_{\bar{x}\in\{-1,1\}^n}x_sx_{t}\det\left(A_sC^\top\right).
\end{align}

The~\cref{eq:uhamcount} gives a $\OO^*(2^n)$ time and polynomial space algorithm in a straightforward manner--for each substitution of the vertex variables from $\{-1,1\}$, we can compute $x_sx_t\det(A_sC^\top)$ in polynomial time, and maintain an accumulating sum as we iterate over all substitutions. 

To summarize, we have shown the following.
\begin{theorem}[\cite{Barvinok}]
	There exists an algorithm, that given an undirected graph $G=(V,E)$ on $n$ vertices, counts the number of Hamiltonian paths in $G$ in time $\OO^*(2^n)$ and polynomial space. 
\end{theorem}
     
We distil all of the above into a recipe that we use for the rest of the problems, where we count (or detect) the number of subgraphs of a specific kind (depending on the problem) in the input graph.

\paragraph{A recipe} 
To compute the number of $st$-Hamiltonian paths in an undirected graph $G$, do the following.   
\subparagraph{Construct a graph $H$.}
To construct $H$, assign weights to the edges of $G$ as follows. Let $X=\{x_v\mid v\in V(G)\}$ be the set of vertex variables. Set the weight of any edge $\{u,v\}$ to be $x_ux_v$. Let $A_s$ be the weighted incidence matrix of $H$ under an arbitrary orientation, with the row corresponding to vertex $s$ removed.   

\subparagraph{Construction of $C$.}
The matrix $C$ is the unweighted incidence matrix of $H$ under the same orientation as in $A_s$, with the row corresponding to $s$ removed.

\subparagraph{Computation.}
The number of $st$-Hamiltonian Paths in $G$ can be computed as follows:
\begin{align*}
\#\mbox{HamPaths}(G,s,t)=\frac{1}{2^n}\sum_{\bar{x}\in\{-1,1\}^n}x_sx_{t}\det\left(A_sC^\top\right)
\end{align*}

We can count the number of Hamiltonian paths in directed graphs as well, using almost the same construction as for undirected graphs; see~\hyperref[app:dirHamP]{Appl. 0.2}.\\
\linebreak
\noindent\customlabel{app:undirBipHamP}{Application 1.1: Detecting Hamiltonian Paths in bipartite undirected graphs}   
Given an undirected bipartite graph $G=(V_1\sqcup V_2,E)$ on $2n$ vertices and two specified vertices $s\in V_1$ and $t\in V_2$, we are interested in deciding whether there is a Hamiltonian path in $G$ from $s$ to $t$. Let $|V_1|=|V_2|=n$.
\begin{figure}
	\centering
	\begin{subfigure}[c]{0.4\textwidth}
		\centering
		\includegraphics[width=0.8\textwidth]{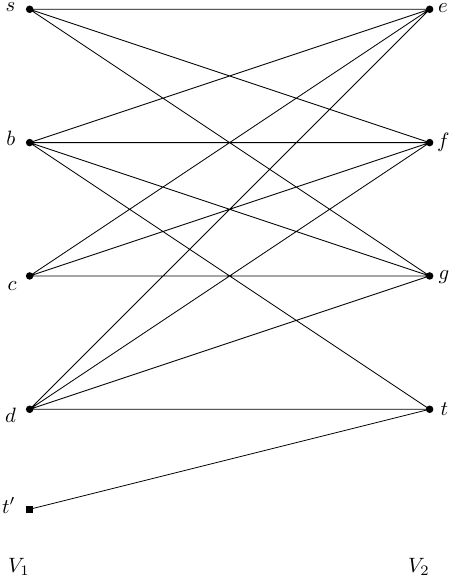}
		\caption{ }
		\label{fig:buhampath}	
	\end{subfigure}
	\begin{subfigure}[c]{0.4\textwidth}
		\centering
		\includegraphics[width=0.8\textwidth]{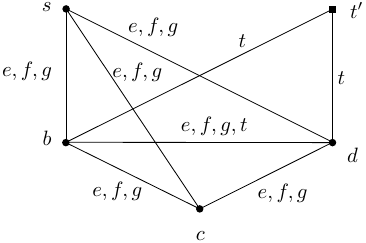}
		\caption{ }
		\label{fig:busqhampath}	
	\end{subfigure}
	\caption{(a) A bipartite undirected graph $G=(V_1\sqcup V_2,E)$(with extra vertex $t'$ attached to $t$). (b) Square graph of $G$, restricted to $V_1$.} 
	\label{fig:BUHamPath}
\end{figure}

\subparagraph{Construction of $H$.}
First, introduce a new vertex $t'$ in $G$ and attach it to $t$. The modified graph is still bipartite, with the parts being $V_1\cup\{t'\}$ and $V_2$. Now the graph $H$ is constructed from the modified $G$ by squaring $G$ as follows. The vertex set $V(H)$ is just $V_1\cup\{t'\}$ and there is an edge between $u$ and $v$ in $H$, labelled $w$, for any $w\in V_2$ such that $uwv$ is a valid path in $G$. Notice that there can be multi-edges in $H$, though they must have distinct labels. See~\cref{fig:BUHamPath} for an example.
Let $X$ be the set of vertex variables, $X=\{x_u|u\in V(H)\}$. Let $z$ be another formal variable.

The weight of an edge $\{u,v\}$ that is labelled $w$ is set to be $z^{W(uwv)}x_ux_v$, where $W(uwv)$ is chosen independently and uniformly at random from $[4|E(H)|]$.
 
\subparagraph{Construction of $C$.}
The matrix $C$ is a $0$-$1$ valued matrix, and its rows are indexed by the vertices in $V_2$. The only non-zero entry in a column indexed by an edge $\{u,v\}$ that is labelled $w$ is in the row $w$. 
\subparagraph{Computation.}
If there exists a Hamiltonian path from $s$ to $t$ in $G$, then the following holds with high probability:
\begin{align}
	\label{eq:buham}
	\sum_{\bar{x}\in\{-1,1\}^{n+1}}x_sx_{t'}\det\left(A_sC^\top\right) \neq0,
\end{align}
where $\bar{x}$ is the $(n+1)$-tuple of distinct vertex variables from $X$, arranged lexicographically. Otherwise, if there does not exist any $s$ to $t$ Hamiltonian path in $G$, then the LHS in~\cref{eq:buham} is always zero.   
\subparagraph{Result.}
The above effectively gives us a randomized algorithm (one sided error) for detecting Hamiltonian paths in undirected bipartite graphs on $n$ vertices, that runs in $\OO^*(2^{n/2})$ time and polynomial space. See~\cref{thm:buham} for a proof.
\subparagraph{Why it works.} The idea is that the edges of any $st$-Hamiltonian path in $G$ can be partitioned into length-two paths that start and end in $V_1$, and one edge that is incident on $t$. The edges in $H$, by construction, correspond to such length-two paths in $G$, with the label of the edge being the middle vertex of the length-two path. Thus an $st'$-Hamiltonian path in $H$, in which all the vertices of $V_2$ appear as labels of the edges, corresponds to an $st$-Hamiltonian paths in $G$. For example, the $st'$-Hamiltonian path in $H$ shown in~\cref{fig:glhampath}, all the vertices of $V_2$ appear as labels, and thus it corresponds to an $st$-Hamiltonian path in $G$. On the other hand, in~\cref{fig:blhampath}, in the highlighted $st'$-Hamiltonian path, label $f$ is missing, and hence it does not correspond to an $st$-Hamiltonian path in $G$. 

The construction of $C$ and the summation over all substitution for $X$ variables from $\{-1,1\}$ ensure that non-$st'$-Hamiltonian spanning trees, as well as `badly' labelled $st'$-Hamiltonian paths, do not contribute to the LHS of~\cref{eq:buham}. Still, distinct `good' labelled $st'$-Hamiltonian paths might cancel each other out. The randomly chosen exponents of $z$ in the edge weights, ensure that at least one survives the summation with high probability by isolation lemma. See~\cref{thm:buham} for a full proof. 
\linebreak
\begin{figure}
	\centering
	\begin{subfigure}[c]{0.3\textwidth}
		\centering
		\includegraphics[width=0.8\textwidth]{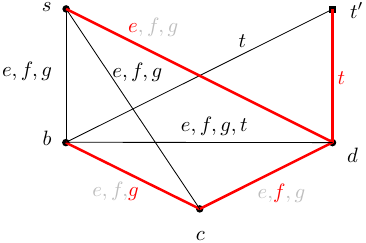}
		\caption{ }
		\label{fig:badhampath}	
	\end{subfigure}
	\begin{subfigure}[c]{0.3\textwidth}
		\centering
		\includegraphics[width=0.8\textwidth]{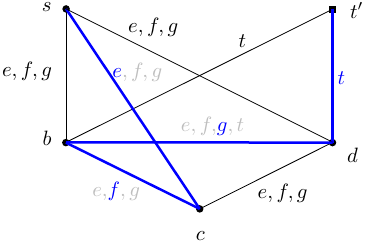}
		\caption{ }
		\label{fig:glhampath}	
	\end{subfigure}
	\begin{subfigure}[c]{0.3\textwidth}
		\centering
		\includegraphics[width=0.8\textwidth]{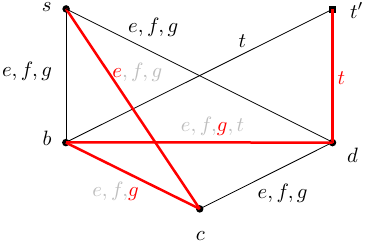}
		\caption{ }
		\label{fig:blhampath}
	\end{subfigure}
	\caption{Spanning trees of the squared graph from~\cref{fig:BUHamPath}. (a) A spanning tree that is not Hamiltonian, with weight $x_s x_b x_c x^3_d {x}_{t'}$. (b) A Hamiltonian spanning tree, where all the $V_2$ labels appear on some edge, with weight $x_s x^2_b x^2_c x^2_d {x}_{t'}$.  (c) A Hamiltonian spanning tree, with weight $x_s x^2_b x^2_c x^2_d {x}_{t'}$, but in which some $V_2$ labels are not covered.} 
	\label{fig:hampathtrees}
\end{figure}

\noindent\customlabel{app:undirIndHamP}{Application 1.1.1: Detecting Hamiltonian paths in undirected graphs with large independent set}

Given an undirected graph $G$ as input, and a maximum independent set $V_0$ of $G$ of size $\alpha(G)$, we can correctly detect whether there exists a Hamiltonian path in $G$ with high probability. We assume that both $s$ and $t$ are in $V\setminus V_0$. We need to slightly modify the construction used in the bipartite undirected graph case as follows.

\subparagraph{Construction of $H$.} The graph $H$ has $V\setminus V_0$ as its vertex set. For each $uwv$ path in $G$, where $w\in V_0$, there is an edge between $u$ and $v$ in $H$, labelled $w$. Also, for each edge $\{u,v\}$ in $G[V\setminus V_0]$, there are $q=|V(H)|-|V_0|-1$ parallel edges, each labelled with a distinct labels from $[q]$. The weight of any edge $\{u,v\}$ in $H$, that is labelled $\ell\in V_0\cup[q]$, is set to be $z^{W(u\ell v)}x_ux_v$. Here $W:E(H)\rightarrow[4|E(H)|]$ is a random weight function on the labelled edges of $H$.
\subparagraph{Construction of $C$.} The first $|V_0|$ many rows of the matrix $C$ are indexed by vertices in $V_0$, and the remaining rows are indexed by numbers in $[q]$. The non-zero entry in the column of $C$ corresponding to a $uv$ edge that is labelled with a vertex $w$ in $V_0$ appears in the row indexed by $w$. For an edge $uv$ labelled with $k\in[q]$, the corresponding column in $C$ has non-zero entry only in row indexed by $k$. Note that all the non-zero entries of $C$ are $1$. 
\subparagraph{Computation.}
The following is true with high probability if there exists an $st$-Hamiltonian path in $G$,
\begin{align}
\label{eq:uiham}
\sum_{\bar{x}\in\{-1,1\}^{|n-\alpha(G)|}}x_sx_t\left(\det(A_sC^\top)\right)~~\neq0,
\end{align}
where $\alpha(G)$ is the size of any maximum independent set of $G$. When there does not exist any $s$ to $t$ Hamiltonian path in $G$, the LHS in~\cref{eq:uiham} is always zero.   
\subparagraph{Result.}
We can correctly detect whether there exists a Hamiltonian path in $G$ on $n$ vertices with high probability (one sided error), given a maximum independent set $V_0$ of $G$, in $\OO^*(2^{n-\alpha(G)})$ time and polynomial space. 
\subparagraph{Why it works.} The idea is that the edges of any $st$-Hamiltonian path in $G$ can be partitioned into length-two paths that start and end in $V\setminus V_0$ with the middle vertex in $V_0$, and into edges with both endpoints in $V\setminus V_0$. The edges in $H$, by construction correspond to such two length paths in $G$, with the label of the edge being the middle vertex of the length-two path, and also to edges in $G$ that have both endpoints in $V\setminus V_0$, each equipped with labels from $[q]$. Thus, an $st$-Hamiltonian path in $H$ in which all the vertices of $V_0$ appear as labels of the edges in the Hamiltonian path corresponds to an $st$-Hamiltonian paths in $G$. The labels from $[q]$ for the copies of edges in $G[V\setminus V_0]$ are there to match the sizes of the matrices $A_s$ and $C$, so that Cauchy-Binet theorem is applicable. 

The construction of $C$ and the summation over all substitution for $X$-variables from $\{-1,1\}$ ensure that non $st$-Hamiltonian path spanning trees, as well as `badly' labelled $st$-Hamiltonian paths, do not contribute to the LHS of~\cref{eq:uiham}. Notice that in any `good' labelled Hamiltonian path of $H$, there are $q$ edges with both endpoints in $V\setminus V_0$, and thus the $q$ distinct labels can cover them perfectly, so that the corresponding set of columns in $C$ is independent.  Still, distinct `good' labelled $st'$-Hamiltonian paths might cancel with each other. The randomly chosen exponents of $z$ in the edge weights ensure that at least one survives the summation with high probability via the isolation lemma,~\cref{lem:isolate}.\\ 

\noindent\customlabel{app:dirBipHamP}{Application 1.2: Detecting Hamiltonian paths in directed bipartite graphs}

As input, we have a directed bipartite graph $G=(V_1\sqcup V_2,E)$, and two specified vertices $s\in V_1$ and $t\in V_2$, such that $|V_1|=|V_2|=n$. We are interested in deciding whether there is a directed Hamiltonian path from $s$ to $t$ in $G$.

The application of the recipe is almost the same as in the undirected bipartite case earlier. We explain only the modifications from the undirected case in the following. The construction of $H$ is the same as before, except that $H$ is a directed graph and there is an edge $(u,v)$ in $H$ labelled $w$ if there is a directed path $uwv$ in $G$. The weight assigned to this edge is $z^{W(uwv)}x_ux^2_v$, where $W:E(H)\rightarrow[4|E(H)|]$ is a random weight function.

The construction of $C$ is the same as before.
In the computation step, we replace the product $x_sx_t$ appearing before the determinant in~\cref{eq:buham} with $x^2_sx_{t'}$, and the order of the roots of unity to be $3$ instead of $2$.
More specifically,
\[
	\sum_{\bar{x}\in\{1,\omega,\omega^2\}^{n+1}}x^2_sx_{t'}\det\left(A_sC^\top\right) \neq 0,
\]
with high probability, if there exists a directed $s$ to $t$ Hamiltonian path in $G$. If there does not exist any directed $s$ to $t$ Hamiltonian path, then the LHS of the above equation evaluates to zero with probability $1$.

Thus, we get a randomized algorithm (one sided error) that correctly detects whether there is a Hamiltonian path in a directed bipartite graph on $n$ vertices, with high probability and runs in time $\OO^*(3^{n/2})$ time and uses polynomial space. See~\cref{thm:bdham} for a proof.\\

\noindent\customlabel{app:dirIndHamP}{Application 1.2.1: A $\OO^*(3^{n-\alpha(G)})$ time algorithm for detecting Hamiltonian path in a directed graph.}

Given a directed graph $G$ as input, and a maximum independent set $V_0$ of $G$ of size $\alpha(G)$, we can detect correctly whether there exists a Hamiltonian path in $G$ with high probability. We assume that both $s$ and $t$ are in $V\setminus V_0$. Again, the construction is very similar to the bipartite case.

\subparagraph{Construction of $H$.} The graph $H$ has $V\setminus V_0$ as its vertex set. For each $uwv$ path in $G$ (with $w\in V_0$) there is an edge from $u$ to $v$ in $H$, labelled $w$,  that is assigned the weight $z^{W(uwv)}x_ux^2_v$. Also, for each edge $uv$ of $G$ such that $u$ and $v$ both lie in $V\setminus V_0$ there are $q = |V(H)|-|V_0|-1$ parallel edges in $H$ from $u$ to $v$, each labelled with a distinct label from $[q]$. The weight of any edge $uv$ in $H$, that is labelled $\ell$ ($\ell\in V_0\cup[q]$) is set to be $z^{W(u\ell v)}x_ux^2_v$. Here $W:E(H)\rightarrow[4|E(H)|]$ is a random weight function on the labelled edges of $H$.
\subparagraph{Construction of $C$.} The first $|V_0|$ rows of the matrix $C$ are indexed by the vertices in $V_0$, and the remaining rows are indexed by numbers in $[q]$. The non-zero entry in the column of $C$ corresponding to an edge $uv$ that is labelled with a vertex $w\in V_0$ is in the row indexed by $w$. For an edge $uv$ labelled with $k\in[q]$, the corresponding column of $C$ has a non-zero entry only in the row indexed by $k$. Note that all the non-zero entries of $C$ are $1$.
\subparagraph{Computation.}
\[
\sum_{\bar{x}\in\{1,\omega,\omega^2\}^{|n-\alpha(G)|}}x^2_sx_t\left(\det(A_sC^\top)\right) \neq 0
\]
with high probability, if there exists a directed $s$ to $t$ Hamiltonian path in $G$. Otherwise, if there does not exist any directed $s$ to $t$ Hamiltonian path in $G$, then LHS in the above equation always evaluates to zero.

\subparagraph{Result.}
Given a directed graph $G$ on $n$ vertices, and an independent set $V_0$ of $G$, we can correctly detect with high probability (one sided error), whether there exists a directed Hamiltonian path in $G$ with high probability, in time $\OO^*(3^{n-\alpha(G)})$ and polynomial space.  

\subparagraph{Why it works.} The idea here is similar to~\hyperref[app:undirIndHamP]{Appl. 1.1.1}. The construction of $C$ is also the same. We can ensure that $\det(A_r[S]C^\top[S])$ is non-zero only for spanning trees $S$ of $H$; that is, $S$ need not be an in-arborescence, since we have traded the in-incidence matrix of $H$ with $C$, for saving variables on the $V_0$ vertices. So, the monomials that appear in $x^2_sx_t\det(A_r[S]C^\top[S])$ (before cancellations), correspond to spanning trees, and not arborescences of $H$. But for each edge $uv$ the factor of $x_ux^2_v$ in its weight, ensures that only the monomials corresponding to directed $st$-Hamiltonian paths in $H$ have $\prod_{v\in V(H)}x^3_v$. And for all the other spanning trees, the monomial corresponding to them in $x^2_sx_t\det(A_r[S]C^\top[S])$ will have at least one variable appearing with degree not divisible by $3$. Because, if the spanning tree is not a Hamiltonian path, then there must be a leaf vertex that is distinct from both $s$ and $t$, and thus the corresponding variable appears with degree either $2$ or $1$. Moreover, even if the spanning tree is a Hamiltonian path but not a directed one, then either there is a vertex variable that appears with degree $4$ or $2$ in the corresponding monomial. Such monomials get cancelled in the summation over all cube-root-of-unity substitution for vertex variables, similar to~\hyperref[app:dirBipHamP]{Appl. 1.2}. 

Again, there could be cancellation amongst monomials corresponding to good Hamiltonian spanning trees, because the signs of the monomials are not all consistent as in the directed spanning tree polynomial. The random weight function $W$ ensures that at least one monomial survives the summation with high probability due to the isolation lemma,~\cref{lem:isolate}. 

%% file: maxmatch.tex
\section{Counting matchings and coverings}\label{sec:match}

\customlabel{app:BipPMCount}{Application 2.1: Counting perfect matchings in bipartite graphs}

As input we have a bipartite graph $G=(V_1\sqcup V_2,E)$, such that $|V_1|=|V_2|=n$. We are interested in counting the number of perfect matchings in $G$. For this counting problem, we illustrate the recipe in action.

\subparagraph{Construction of $H$.} The graph $H$ is constructed from $G$ by introducing a special vertex $s$ and adding an edge from it to every vertex in $V_1$. Let $X$ be the set of vertex variable for vertices in $V_1$, i.e., $X=\{x_u\mid u\in V_1\}$, and let $y$ another formal variable. The weights of the edges in $H$ are defined as follows. For an edge $\{u,v\}$, $u\in V_1$ and $v\in V_2$, the weight is $x_u$, and for an edge from $s$ to a vertex in $V_1$, the weight is $y$. Let $A_s$ be the weighted incidence matrix of $H$ under an arbitrary orientation, with the row corresponding to $s$ removed. 
\subparagraph{Construction of $C$.} The matrix $C$ set to be the incidence matrix of the unweighted incidence matrix of $H$ under the same orientation as in $A_s$, with row corresponding to the vertex $s$ removed. 

\subparagraph{Computation.}  

The number of perfect matchings in $G$ can be computed as follows.

\[
\#\mathrm{PM}(G) = \frac{1}{2^n}\sum_{\bar{x}\in\{-1,1\}^{n}}[y^n]\left(\bold{x}\det(A_sC^\top)\right)
\] 
where, $\bold{x} = \prod\limits_{x_u\in X}x_u$. 

\subparagraph{Result.}
We can compute the number of perfect matchings in $G$ (on $2n$ vertices) in $\OO^*(2^n)$ time and polynomial space. See~\cref{thm:bpm} for a proof. 

\subparagraph{Why it works.} The idea is that a spanning tree of $H$ that includes all the edges incident on $s$ and in which the vertices of $V_1$ have degree $2$ contains a perfect matching of $G$. Moreover, any perfect matching of $G$ can be extended to such a spanning tree of $H$ in a unique way. See~\cref{fig:PM} for an illustrative example.
\linebreak

\begin{figure}
	\centering
	\begin{subfigure}{0.4\textwidth}
		\centering
		\includegraphics[width=0.8\textwidth]{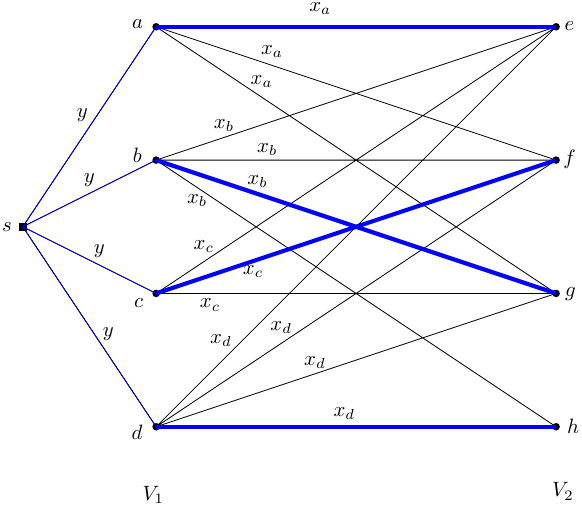}
		\caption{ }
		\label{fig:goodTree}	
	\end{subfigure}
	\begin{subfigure}{0.4\textwidth}
		\centering
		\includegraphics[width=0.8\textwidth]{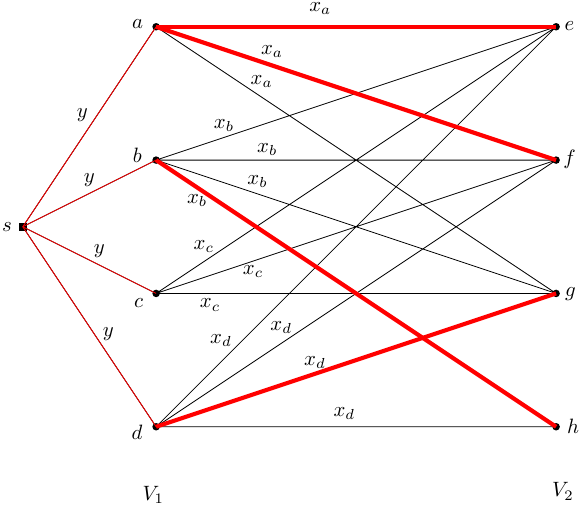}
		\caption{ }
		\label{fig:badTree}	
	\end{subfigure}
	\caption{Matchings in a bipartite graph $G=(V_1\sqcup V_2,E)$ on $8$ vertices. (a) A `good' spanning tree containing a matching (graph edges are thick blue edges). The monomial corresponding to the tree is $y^4x_ax_bx_cx_d$. (b) A `bad' spanning tree; does not contain any matching. Monomial corresponding to the tree is $y^4x^2_ax_bx_d$.}
	\label{fig:PM}
\end{figure}

\noindent\customlabel{app:MMCount}{Application 2.2: Counting Maximum Matchings}

For a graph $G=(V,E)$, let $\nu=\nu(G)$ denote the size (number of edges) of the maximum matching in $G$. Denote by $\#\mbox{MM}(G)$ the number of maximum matchings in $G$. 
First, we compute the Gallai-Edmonds Decomposition of $G$ in polynomial time.
Let $A,C,D$ be the parts of the Gallai-Edmonds Decomposition of $G$, and let the connected components of $D$ be $\{D_i\}_{i\in[k]}$ for some $k$.

\subparagraph{Construction of $H$.} 

First, we define a weighted bipartite graph $\widetilde H=(A\sqcup[k],\widetilde E)$, such that there is an edge between $a\in A$ and $i\in [k]$ in $\widetilde H$ if there exists a vertex $d\in D_i$ that is adjacent to $a$ in $G$. That is, $\widetilde E = \{ \{a,i\} \mid a\in A, i\in [k], \exists d\in D_i, \{a,d\}\in E\}$. Weight of an edge $\{a,i\}$ of $\widetilde H$, $w(\{a,i\})$, is given as follows.
\[
	w(\{a,i\})=\sum_{d\in D_i:(a,d)\in E}\#\mbox{PM}(D_i\setminus \{d\})
\] 
For any $i\in [k]$, let 
\[
	W_i = \sum_{d\in D_i}\#\mbox{PM}(D_i\setminus \{d\})
\]
where $\#\mbox{PM}(D_i\setminus d)$ denotes the number of perfect matchings in the graph $G[D_i\setminus \{d\}]$. 

Now, construct $H$ from $\widetilde H$ as follows. Introduce a special vertex $s$ and add edges from it to all the vertices in $A\sqcup [k]$. So, $H= (\{s\}\cup A\cup [k], \widetilde E\cup \{\{s,t\}\mid t\in A\cup[k]\})$.
Let $X$ be the sets of vertex variables for vertices in $A$, i.e., $X=\{x_a \mid a\in A\}$  and let $Y = \{y_1,y_2\}$ be another set of formal variables.
The edges in $H$ are assigned weights as follows: (i) for an edge $\{a,i\}$ in $\widetilde E$, the weight is $w(\{a,i\})x_a$, (ii) for an edge from $s$ to a vertex in $A$, the weight is $y_1$, (iii) and for an edge from $s$ to any vertex in $[k]$, its weight is $W_iy_2$.

\subparagraph{Construction of $C$.} The matrix $C$ set to be the incidence matrix of the unweighted graph underlying $H$, with row corresponding to $s$ removed. 

\subparagraph{Computation.}  

The number of maximum matchings in $G$ can be computed as follows.

\[
\#\mbox{MM}(G) = \#\mbox{PM}(C(G))\cdot  \frac{1}{2^{|A|}}\sum_{\bar{x}\in\{-1,1\}^{|A|}}[y_1^{|A|}y_2^{k-|A|}]\left(\bold{x}\det(A_sC^\top)\right)
\] 
where, $\bold{x} = \prod_{x_u\in X}x_u$. 

Notice that we can compute $\#\mbox{PM}(C(G))$ in $\mcO^*(2^{|C|/2})$ time using \cite{Bjorklund11} as a black box (in case of bipartite graphs, we can use our construction for counting perfect matching instead). In the construction of $H$, we need to compute $\#\mbox{PM}(D_i\setminus d)$ for all the components $D_i$ and all $d\in D_i$. These computations can again be performed in $\mcO^*(2^{(|D_i|-1)/2})$ time using \cite{Bjorklund11}. Finally, computing the above sum takes $\mcO^*(2^{|A|})$ time. Overall, the complexity is 
\[
	\mcO^*(2^{|C|/2}+\sum_i2^{(|D_i|-1)/2}+2^{|A|})
\] which is in turn $\mcO^*(2^{\nu})$ due to \cref{item:5} of the Gallai-Edmonds Decomposition.
\subparagraph{Result.}
We can count the number of maximum matchings in $G$ in time $\mcO^*(2^{\nu(G)})$ and polynomial space, where $\nu(G)$ is the size of any maximum matching in $G$. See~\cref{thm:maxmatch} for a proof.

\subparagraph{Why it works.} From \cref{item:4}, it is clear that the number of maximum matchings in $G$, $\#\mbox{MM}(G)$ is equal to $\#\mbox{PM}(C(G))$ times the number of maximum matchings in $A\cup D$. Moreover,  every spanning tree of $H$ that contains all the edges between $s$ and vertices in $A$, and exactly $k-|A|$ edges between $s$ and vertices in $[k]$ such that $A$ vertices have degree $2$, contains a maximum matching of $\widetilde H$, and vice versa. The weights $w(a,i)$ and $w(D_i)$ are multiplied to the edges to account for the number of ways a maximum matching of $\widetilde H$ can be extended to that of $A\cup D$.\\

\noindent\customlabel{app:kStarCount}{Application 2.3: Counting perfect $k$-star covers}

As input we have an undirected graph $G=(V,E)$. We are interested in counting the number of perfect $k$-star covers of $G$. A $k$-star is a tree on $k$ vertices where one vertex, called the centre of the star, has degree $k-1$, and all other vertices are leaves. In particular, a matching edge is a $2$-star. A perfect cover of $G$ with $k$-stars is a set of vertex-disjoint subgraphs, each of them a $k$-star, that together cover all the vertices of $G$.  

We do this count in the following way: 
\begin{enumerate}
	\item Select a subset $V_0\subseteq V$ of size $n/k$, and count the number of $k$-star covers of $G$ in which, for every $k$-star in the cover, its centre lies in $V_0$. Denote the number of such $k$-star covers by $\#\mbox{k-Star}_{V_0}(G)$.
	\item Then, the number of $k$-star covers of $G$ is: 
	\[
		\#\mbox{k-Star}(G)=\sum\limits_{V_0\in\binom V {n/k}}\#\mbox{k-Star}_{V_0}(G).
	\]
\end{enumerate}

We show how to compute $\#\mbox{k-Star}_{V_0}(G)$ for a fixed $V_0\subset V$.
\subparagraph{Construction of $H$.} The graph $H$ is constructed from $G$ by introducing a special vertex $s$, and adding an edge from it to every vertex in $V_0$. The only edges of $G$ that are present in $H$ are those with one end point in $V_0$ and the other in $V\setminus V_0$. Let $X=\{x_u\mid u\in V_0\}$ be the set of vertex variables for vertices in $V_0$. Let $y$ be another formal variable. The weights of the edges in $H$ are given as follows. For an edge $\{u,v\}$ of $H$, $u\in V_0$ and $v\in V\setminus V_0$, the weight is $x_u$, and for an edge from $s$ to any vertex of $V_0$, the weight is $y$.
\subparagraph{Construction of $C$.} The matrix $C$ set to be the incidence matrix of the unweighted graph underlying $H$, with row corresponding to $s$ removed. 

\subparagraph{Computation.}  
The number of $k$-star covers of $G$, where the centres of the stars lie in $V_0$, can be computed as follows.
\[
\#\mbox{k-Star}_{V_0}(G) = \frac{1}{k^{n/k}}\sum_{\bar{x}\in\{1,\omega,\ldots,w^{k-1}\}^{n/k}}[y^{n/k}]\left(\bold{x}\det(A_sC^\top)\right)
\] 
where $\bold{x} = \prod_{x_u\in X}x_u$ and $1,\omega,\ldots,w^{k-1}$ are the $k$th roots of unity.

\subparagraph{Why it works.} The idea is that a spanning tree of $H$ that contains all the edges that are incident on $s$ and has the degree of vertices in $V_0$ as $k$ corresponds to a perfect $k$-star cover of $G$ with the centres of the stars being in $V_0$. Moreover, any perfect $k$-star cover of $G$, where only the vertices in $V_0$ are to be the centres of the stars, can be extended to such a spanning tree of $H$ in a unique way. 

Next, we need to sum up over all $n/k$ size subsets $V_0$ of $V$. The overall complexity is then $\mcO^*\left(\binom n {n/k}k^{n/k}\right)$ which is $\mcO^*(2^{nH(1/k)+(n\log k)/k})$ where $H$ is the binary entropy function. Here, 
\[
	H(1/k)=(1/k)\log k -(1-1/k)\log(1-1/k)\le (2/k)\log k
\] 
when $k\ge 2$. Thus the overall complexity is $\mcO^*(2^{3(n\log k) /k})$.

\subparagraph{Result.}
We can compute the number of perfect $k$-star covers of a graph $G$ on $n$ vertices in $\mcO^*(2^{3(n\log k) /k})$ time, and polynomial space. See~\cref{thm:kstar} for a proof. 

\begin{remark}
	Some observations:
\begin{enumerate}
	\item The number of perfect $k$-star covers of a graph on $n$ vertices for $k=\mcO( n^{\delta})$ for $\delta\in (0,1)$ can be computed in $\mcO^*\left(2^{\mcO(n^{1-\delta}\log n)}\right)$, i.e., $\mcO^*(2^{o(n)})$ time.
	\item For $k=\Theta(n)$, the number of perfect $k$-star covers of a graph on $n$ vertices can be computed in $\poly(n)$ time.
	\item Notice that $\lim\limits_{k\rightarrow \infty}2^{(3\log k)/k}=1$. Hence, the running time approaches an arbitrarily small exponential function as $k$ grows larger.
\end{enumerate}
\end{remark}

\noindent\customlabel{app:BipkMCount}{Application 2.4: Counting $k$-Matchings in Bipartite Graphs}

As input, we have a bipartite graph $G=(V_1\sqcup V_2,E)$ such that $|V_1|=n_1,|V_2|=n_2$, $n_1\le n_2$, and $n_1+n_2=n$. We are interested in counting the number of matchings of size $k$ in $G$, denoted $\#\mbox{M}_k(G)$, using the previous recipe.
\subparagraph{Construction of $H$.} The graph $H$ is constructed from $G$, by introducing a special vertex $s$, and adding an edge from it to every vertex in $V_1\cup V_2$. The variable sets are, $ X=\{x_u: u\in V_1\}$ and $ Y=\{y_1,y_2\}$. The weights of the edges in $H$ are as follows. For an edge $\{u,v\}$, $u\in V_1, v\in V_2$, the weight is $x_u$, and for an edge from $s$ to any $V_i$ vertex, the weight is $y_i$ for $i\in\{1,2\}$. 
\subparagraph{Construction of $C$.} The matrix $C$ is set to be the incidence matrix of the unweighted graph underlying $H$, with row corresponding to $s$ removed. 

\subparagraph{Computation.}  

The number of matchings of size $k$ in $G$ can be computed as follows.

\[
\#\mbox{M}_k(G) = \frac{1}{2^{n_1}}\sum_{\bar{x}\in\{-1,1\}^{n_1}}[y_1^{n_1}y_2^{n_2-k}]\left(E_k( X)\det(A_sC^\top)\right)
\] 
where, $E_k( X)$ is the elementary symmetric polynomial of degree $k$ over the variables $ X$.

Notice that $E_k(X)$ can be efficiently computed as $E_k(X)=[t^k]\Pi_{u\in V_1}(1+t\cdot x_u)$ where $t$ is a variable. The above gives a $\mcO^*(2^{n/2})$ algorithm for counting $k$-matchings for any $k$.

\subparagraph{Why it works.} The idea is that a spanning tree of $H$ that contains all the edges between $s$ and $V_1$, and exactly $n_2-k$ edges between $s$ and $V_2$, such that exactly $k$ vertices of $V_1$ have degree $2$, contains a $k$ size matching of $G$, and vice versa.

\subparagraph{Result.}
We can count the number of $k$-matchings in a bipartite graph $G$, in time $\mcO^*(2^{n/2})$ and polynomial space, where $n$ is the number of vertices in $G$. See~\cref{thm:kmatch} for a proof.

%% file: thmproofs.tex
\section{Formal statements and their proofs}\label{sec:thmproofs}
In this section, we give formal proofs of the claimed results in the applications.

\subsection*{Proof for \hyperref[app:undirBipHamP]{Application 1.1}}
\begin{theorem}[\cite{ Bjorklund14}]
	\label{thm:buham}
		Given an undirected bipartite graph $G$ on $n$ vertices, we can correctly detect whether a Hamiltonian path exists in $G$ with high probability in time $O^*(2^{n/2})$ and polynomial space. 
\end{theorem} 
\begin{proof}
	For notational convenience, we assume the graph has $2n+1$ vertices (the $+1$ is because of $t'$). Using Cauchy-Binet theorem, we can write 
	\[
		\det(A_sC^\top) = \sum_{\substack{F\subseteq E(H)\\|F|=n}}\det(A_s[F])\det(C[F]).
	\]
	Columns of $A_s$ and $C$ are indexed by the labelled edges $E(H)$. Let $F$ be a set of $n$ such edges. From the proof of ~\cref{lem:mttU} (part (b)), if $F$ is a spanning tree of $H$ then $\det(A_s[F])$ is $\pm1$ times the monomial corresponding to that tree, where $x_{e=uwv}=x_ux_vz^{W(uwv)}$; otherwise $\det(A_s[F])=0$. Moreover, $\det(C[F])=\pm1$ if all the vertices in $V_2$ are picked exactly once by some labelled edge in $F$; otherwise $\det(C[F])=0$, since one of the rows of $C[F]$ becomes an all zero row. 
	
	Therefore $\det\left(A_sC^\top\right)$ is a signed sum of spanning trees of $H$ in which every vertex $v\in V_2$ appears as the label of exactly one edge in the tree. In other words, $\det\left(A_sC^\top\right)$ contains monomials corresponding to spanning trees of $G$ in which all vertices of $V_2$ have degree $2$. 
	
	Let $T$ be such a spanning tree of $G$ with its corresponding monomial being $T(X,z)$. Observe that if $T$ has any odd degree vertex in $V_1$ other than $s$ or $t'$, then 
	\[
		\sum_{\bar{x}\in\{-1,1\}^{n+1}}x_sx_{t'}T(X,z)=0.
	\] 
	
	Moreover if $s$ and $t'$ are the only leaf vertices then $T$ is a Hamiltonian path between $s$ and $t'$ (equivalently $t$), and
	 \[
	 	\sum_{\bar{x}\in\{-1,1\}^{n+1}}x_sx_{t'}T(X,z)=2^nz^{W(T)}
	 \] 
	 where $W(T)=\sum_{uwv\in T}W(uwv)$. By~\cref{lem:isolate}, $W$ isolates a min-weight $st'$ Hamiltonian path with probability $3/4$. Therefore, if $G$ has a $s,t$ Hamiltonian path then with probability $3/4$, 
	 \[					\sum_{\bar{x}\in\{-1,1\}^n}x_sx_{t'}\det\left(A_sC^\top\right)~~\neq0,
	 \] and if $G$ has no $s,t$ Hamiltonian path then the sum is $0$. 
	 
	 For the runtime, observe that under each substitution of $\bar{x}$ from $\{-1,1\}^{n+1}$, we can compute the determinant of $A_sC^\top$, i.e., determinant of a matrix whose entries are univariate polynomials ( in $z$) of degree at most $4|E(H)|(n+1)$, which can be computed in polynomial time~\cite[Theorem 3.2]{AAM}. Since, there are $2^{n+1}$ many determinants to compute, we will be done in $O^*(2^{n})$ time. We require only polynomial space to accumulate the outer sum over all the substitutions.
\end{proof} 

\subsection*{Proof for \hyperref[app:dirBipHamP]{Application 1.2}}
\begin{theorem}[\cite{BKK17}]
	\label{thm:bdham}
	Given a directed bipartite graph $G$ on $n$ vertices, we can correctly detect whether there exists a directed Hamiltonian path in $G$ with high probability in $O^*(3^{n/2})$ time and polynomial space. 
\end{theorem} 
\begin{proof}
	The proof is very similar to the undirected case. The columns of $A_s$ and $C$
	are indexed by the labelled arcs $E(H)$. Let $F$ be a set of $n$ such arcs.
	If $F$ is a spanning tree of $H$ in the underlying undirected setting, then
	$\det(A_s[F])$ is $\pm 1$ times the monomial corresponding to the tree, where
	$x_{e=uwv} = x_u x_v^2 z^{W(uwv)}$; otherwise, $\det(A_s[F]) = 0$. As in the
	undirected case, $\det(C[F]) = \pm 1$ if all vertices in $V_2$ are picked
	exactly once by some labelled arc in $F$; otherwise, $\det(C[F]) = 0$.
	
	In other words, $\det(A_s C^\top)$ is a signed sum of spanning trees (in the
	underlying undirected graph) of $G$ in which all vertices of $V_2$ have both
	in-degree and out-degree equal to $1$. Let $T$ be such a spanning tree of $G$,
	with corresponding monomial $T(X,z)$. Observe that if $T$ has a leaf
	$u \in V_1$ that is not $s$ or $t'$, then the degree of $x_u$ in the monomial
	$T(X,z)$ is either $1$ or $2$, depending on the arc of $T$ incident to $u$.
	Therefore, if $T$ has any leaf vertices in $V_1$ other than $s$ and $t'$, then
	\[
	\sum_{\bar{x} \in \{1,\omega,\omega^2\}^{n+1}} x_s^2 x_{t'}\, T(X,z) = 0.
	\]
	Thus, the only trees for which this sum is non-zero are those in which $s$ and
	$t'$ are the only leaves. This occurs precisely when $T$ is a path between $s$
	and $t'$ in the underlying undirected version of $H$.
	
	Now suppose that $T$ is such an undirected path, but
	is not a directed path from $s$ to $t'$. Then there exists a vertex
	$u \in V_1$ with in-degree $2$ and out-degree $0$ in $T$, in which case the
	degree of $x_u$ in $T$ is $4$; or with in-degree $0$ and out-degree $2$, in
	which case the degree of $x_u$ in $T$ is $2$. In either situation, the degree
	of $x_u$ in $T(X,z)$ is not a multiple of $3$, and therefore
	\[
	\sum_{\bar{x} \in \{1,\omega,\omega^2\}^{n+1}} x_s^2 x_{t'}\, T(X,z) = 0.
	\]
	
	If $T$ is a Hamiltonian path from $s$ to $t'$, then
	\[
	\sum_{\bar{x} \in \{1,\omega,\omega^2\}^{n+1}} x_s^2 x_{t'}\, T(X,z)
	= 3^{n+1} z^{W(T)},
	\]
	where $W(T)$ is the total weight of the arcs in $T$. By the Isolation Lemma,~\cref{lem:isolate},
	$W$ isolates a minimum-weight Hamiltonian $s \to t'$ path with probability
	$3/4$. Therefore, if $G$ has an $s$-to-$t$ Hamiltonian path, then with
	probability $3/4$,
	\[
	\sum_{\bar{x} \in \{1,\omega,\omega^2\}^{n+1}}
	x_s^2 x_{t'}\, \det(A_s C^\top) \neq 0,
	\]
	and if $G$ has no $s$-to-$t$ Hamiltonian path, then the sum is $0$.
	
\end{proof}

\subsection*{Proof for \hyperref[app:BipPMCount]{Application 2.1}}
\begin{theorem}[\cite{Ryser}]
	\label{thm:bpm}
	Given a bipartite graph $G=(V_1\sqcup V_2,E) $ on $n$ vertices, the number of perfect matchings in $G$ can be computed in $O^*(2^{n/2})$ time and polynomial space. 
\end{theorem}
\begin{proof}
	Let $\mathrm{sp}(X,y) = \det(A_s C^\top)$ be the spanning-tree polynomial of $H$.
	Then every monomial of $[y^n]\mathrm{sp}(X,y)$ corresponds to a spanning tree of
	$H$ in which all edges between $s$ and $V_1$ are included. Let $T$ be such a
	tree, and let $T(X,y)$ be its corresponding monomial. Observe that if any vertex
	$u \in V_1$ has degree $1$, i.e., the only edge incident to it is $(s,u)$, then
	\[
	\sum_{\bar{x} \in \{-1,1\}^{n}} \mathbf{x}\, [y^n](T(X,y)) = 0.
	\]
	
	Moreover, if all vertices in $V_1$ have degree at least $2$ in $T$, then every
	vertex in $V_1$ is forced to have degree exactly $2$; in this case, $T$
	contains a perfect matching of $G$. Also, every perfect matching of $G$ can be
	extended to such a tree in a unique way by adding all edges from $s$ to $V_1$. For such a tree
	$T$, we have
	\[
	\sum_{\bar{x} \in \{-1,1\}^{n}} \mathbf{x}\, [y^n](T(X,y)) = 2^n.
	\]
	
	Thus,
	\[
	\sum_{\bar{x} \in \{-1,1\}^{n}} [y^n]\!\left(\mathbf{x}\,\mathrm{sp}(X,y)\right)
	= 2^n \times (\text{number of perfect matchings in } G).
	\]
	
\end{proof}

\subsection*{Proof for \hyperref[app:BipkMCount]{Application 2.4}}
\begin{theorem}
	\label{thm:kmatch}
	Given a  graph $G=(V_1\sqcup V_2,E) $ on $n$ vertices, the number of $k$-matchings in $G$ can be computed in $\mcO^*(2^{n/2})$ time and polynomial space. 
\end{theorem}
\begin{proof}
	The proof of this is very similar to the perfect matching case. Again, let
	$\mathrm{sp}(X, Y = \{y_1, y_2\}) = \det(A_s C^\top)$ be the spanning tree
	polynomial of $H$. Then every monomial of
	$[y_1^{n_1} y_2^{\,n_2 - k}]\mathrm{sp}(X, y_1, y_2)$ corresponds to a spanning
	tree of $H$ in which all edges between $s$ and $V_1$, and exactly $n_2 - k$
	edges between $s$ and $V_2$, are included. Let $T$ be such a tree, and let
	$T(X, y_1,y_2)$ denote its corresponding monomial.
	
	Let $S$ be a subset of $V_1$ of size $k$, and let $x_S = \prod_{u \in S} x_u$
	be the corresponding monomial of $E_k(X)$. Observe that if any vertex in $S$ is
	a leaf of $T$, then
	\[
	\sum_{\bar{x} \in \{-1,1\}^{n}} x_S\, [y_1^{n_1} y_2^{\,n_2 - k}](T(X,y_1,y_2))
	= 0.
	\]
	Moreover, if none of the vertices in $S$ is a leaf in $T$, then $T$ contains a
	matching of size $k$ of $G$ that saturates $S$. Note that any $k$-matching of
	$G$ that saturates $S$ can be extended to such a tree by adding all edges
	between $s$ and $V_1$, and between $s$ and the unmatched vertices in $V_2$.
	For such trees, we have
	\[
	\sum_{\bar{x} \in \{-1,1\}^{n}} x_S\, [y_1^{n_1} y_2^{\,n_2 - k}](T(X,y_1,y_2))
	= 2^{n_1} \times (\text{number of $k$-matchings of $G$ saturating } S).
	\]
	
	Thus,
	\[
	\sum_{\bar{x} \in \{-1,1\}^{n_1}}
	[y_1^{n_1} y_2^{\,n_2 - k}]\left(E_k(X)\det(A_s C^\top)\right)
	= 2^{n_1} \times (\text{number of $k$-matchings of } G).
	\]	
\end{proof}

\subsection*{Proof for \hyperref[app:MMCount]{Application 2.2}}
\begin{theorem}
	\label{thm:maxmatch}
	Given a  graph $G=(V_1\sqcup V_2,E) $ on $n$ vertices, the number of maximum matchings in $G$ can be computed in $\mcO^*(2^{\nu(G)})$ time and polynomial space. 
\end{theorem}
\begin{proof}
	From the Gallai-Edmonds Decomposition (see \cref{item:4}), it is clear that
	the number of maximum matchings in $G$, $\#\mathrm{MM}(G)$, is equal to $\#\mathrm{PM}(C(G))$ times the number of maximum matchings
	in $A \cup D$. Suppose $M$ is a maximum matching of $\widetilde H$. Then, using
	the properties of the Gallai-Edmonds Decomposition, we know that $M$ can be
	extended to a maximum matching of $A \cup D$ in
	\[
	w(M)
	:= \prod_{(a,i) \in M} w(a,i)
	\cdot \!\!\!\!\prod_{\substack{i \in [k] \\ \text{unmatched in } M}}
	\!\!\!\! W_i
	\]
	ways. Moreover, any maximum matching of $\widetilde H$ is a matching of size
	$|A|$ in $\widetilde H$.
	
	Now we mimic the proof of \cref{thm:kmatch}. Let $T$ be a spanning tree of $H$
	that contains all edges from $s$ to $A$, and exactly $k - |A|$ edges between $s$
	and $[k]$. Clearly, if any vertex in $A$ is a leaf of $T$, then
	\[
	\sum_{\bar{x} \in \{-1,1\}^{n}}
	\mathbf{x}\, [y_1^{|A|} y_2^{\,k-|A|}](T(X,y_1,y_2)) = 0.
	\]
	Therefore, let us assume that $T$ is a tree with all leaves on the $[k]$ side.
	Such a tree $T$ contains a maximum matching of $\widetilde H$, say $M$. For
	each vertex $i \in [k]$ that is incident to $s$ in $T$ (i.e., unmatched in
	$M$), the corresponding edge $(s,i)$ contributes a factor of $W_i$ multiplied to the
	$y_2$ variable in the weight of $(s,i)$. The edges $(a,i)$, $a\in A$ and
	$i\in[k]$, in $T$ (i.e., the matching edges) contribute the factor $w(a,i)x_a$.
	
	Therefore, for such trees $T$, we have
	\begin{align*}
	\sum_{\bar{x} \in \{-1,1\}^{n}}
	\mathbf{x}\, [y_1^{|A|} y_2^{\,k-|A|}](T(X,y_1,y_2))
	= 2^{|A|}
	\times &\left(\substack{\text{number of ways $T$ can be}\\\text{extended to a perfect matching of }}
	\widetilde H\right).
	\end{align*}
	Also, note that every perfect matching corresponds to some tree in $H$.
	
	Thus, we have shown that
	\[
	\sum_{\bar{x} \in \{-1,1\}^{|A|}}
	[y_1^{|A|} y_2^{\,k-|A|}]
	\left(\mathbf{x}\, \det(A_s C^\top)\right)
	= 2^{|A|}
	\times \#\text{MM}(G[A\cup D]).
	\]
	
\end{proof}

\subsection*{Proof for \hyperref[app:kStarCount]{Application 2.3}}
We need the folllowing lemma for the proof.
\begin{lemma}\label{lemma:roots}
	Let $\omega$ be a primitive $k$-th root of unity. Then for any integer $i$
	with $1\le i\le k$,
	\[
	\sum_{x\in\{1,\omega,\omega^2,\dots,\omega^{k-1}\}} x^i
	=
	\begin{cases}
		0, & \text{if } i<k,\\[4pt]
		k, & \text{if } i=k.
	\end{cases}
	\]
\end{lemma}

\begin{proof}
	If $i=k$ then $x^i=(\omega^j)^k=(\omega^k)^j=1$ for every term, so the sum is $k$.
	
	If $1\le i<k$ then $\omega^i\neq 1$, and the sum is a geometric series:
	\[
	\sum_{j=0}^{k-1}\omega^{ij}=\frac{1-(\omega^{i})^{k}}{1-\omega^{i}}
	=\frac{1-1}{1-\omega^{i}}=0,
	\]
	since $(\omega^{i})^{k}=\omega^{ik}=(\omega^k)^i=1$ and the denominator in non-zero since $\omega$ is primitive.
\end{proof}

\begin{theorem}\label{thm:kstar}
	Given an undirected graph $G=(V,E) $ on $n$ vertices, the number of $k$-star coverings in $G$ can be computed in $\OO^*\left(\left(2^{(3\log k) /k}\right)^n\right)$ time and polynomial space. 
\end{theorem}
\begin{proof}
	The proof of this is very similar to the bipartite perfect matching case. Let
	$\mathrm{sp}(X,y)=\det(A_s C^\top)$ be the spanning-tree polynomial of $H$.
	Then every monomial of $[y^{n/k}]\mathrm{sp}(X,y)$ corresponds to a spanning
	tree of $H$ such that all the edges between $s$ and $V_0$ are taken by the tree.
	Let $T$ be such a tree, and let its corresponding monomial be $T(X,y)$. Observe
	from \cref{lemma:roots} that if any of the vertices $u \in V_0$ has degree
	$< k$, i.e., degree $< k-1$ excluding the edge incident on $s$, then
	\[
	\sum_{\bar{x} \in \{1,\omega,\ldots,\omega^{k-1}\}^{n/k}}
	\mathbf{x}\,[y^{n/k}](T(X,y)) = 0.
	\]
	
	Therefore, we can assume that in all the trees that survive the sum, all the
	vertices in $V_0$ have degree exactly $k$, in which case $T$ contains a perfect
	$k$-star cover of $G$ with $V_0$ being the set of centres of the stars. Also,
	every perfect $k$-star cover of $G$ with centres allowed from $V_0$ alone, can be extended to such a tree by
	adding all the edges from $s$ to $V_0$. For such a tree $T$, we have
	\[
	\sum_{\bar{x} \in \{1,\omega,\ldots,\omega^{k-1}\}^{n/k}}
	\mathbf{x}\,[y^{n/k}](T(X,y)) = k^{n/k}.
	\]
	
	Thus,
	\[
	\sum_{\bar{x} \in \{1,\omega,\ldots,\omega^{k-1}\}^{n/k}}
	[y^{n/k}]\left(\mathbf{x}\,\mathrm{sp}(X,y)\right)
	= k^{n/k}
	\]
	times the number of $k$-star covers of $G$ where $V_0$ is the set of centre vertices of the stars.
	
\end{proof}
\begin{corollary}
	For all $\epsilon>0$, there exists some $K\in\bbN$ such that $\forall k>K$, the number of $k$-star covers can be counted in $\mcO^*((1+\epsilon)^n)$ time in graphs on $n$ vertices.
\end{corollary}

%% file: dirHamPcount.tex
\noindent\customlabel{app:dirHamP}{Application 0.2: Counting Hamiltonian Paths in directed graphs}

As input, we have a directed graph $G=(V,E)$, and two specified vertices $s ,t\in V$. We are again interested in counting the number of directed hamiltonian paths from $s$ to $t$ in $G$.

The application of the recipe is almost same as the in the undirected case. We explain only the modifications from the undirected case in the following. The construction of $H$ is same as before, except that $H$ is a directed graph. The weight assigned to an edge $(u,v)$ is again $x_ux_v$.

The construction of $C$ is different from that in the undirected case. $C$ is now the in-incidence matrix, where columns are labelled by edges, and rows by vertex, except that the row corresponding to $s$ is removed. The column $(u,v)$ has $1$ in the row $v$, and $0$ everywhere else.

Then we can count the number of directed hamiltonian paths in $G$ from $s$ to $t$ as follows:

\[
\#\mbox{HamPaths}(G,s,t)=\frac{1}{2^n}\sum_{\bar{x}\in\{-1,1\}^n}x_sx_{t}\det\left(A_sC^\top\right)
\]


\begin{theorem}\cite{HeldKarp}
		Given a directed bipartite graph $G$ on $n$ vertices, we can counts the number of directed hamiltonian paths in $G$ in time $O^*(2^n)$ and polynomial space. 
\end{theorem}

\begin{proof}
	The proof is very similar to the undirected case. Let
	$p(X) = \mathrm{Arb}_s(X) = \det(A_s C^\top)$ be the $s$-arborescence
	polynomial, where $x_e = x_u x_v$. Consider a monomial of $p$,
	$x_s x_t T(X)$, corresponding to an $s$-arborescence $T$. Again, if any
	variable $x_u$ has odd degree in $x_s x_t T(X)$, then
	\[
	\sum_{\bar{x} \in \{-1,1\}^n} x_s x_t T(X) = 0,
	\]
	and otherwise,
	\[
	\sum_{\bar{x} \in \{-1,1\}^n} x_s x_t T(X) = 2^n.
	\]
	Moreover, every directed Hamiltonian path from $s$ to $t$ is an
	$s$-arborescence with only $s$ and $t$ as the two leaves, and vice versa.
	Therefore,
	\[
	\sum_{\bar{x} \in \{-1,1\}^n} x_s x_t p(X)
	\]
	is $2^n$ times $\#\mathrm{HamPaths}(G,s,t)$.
	
\end{proof}

%% file: concl.tex
\section{Conclusion}\label{sec:concl}
We see our hamiltonian path algorithm as unifying the algorithms for detecting hamiltonian cycles in undirected bipartite~\cite{Bjorklund14} and directed bipartite graphs~\cite{BKK17}.
It remains to be seen whether our technique for finding hamiltonian paths in 
undirected bipartite graphs can be made amenable to general undirected graphs, 
like Bj\"{o}rklund's algorithm~\cite{Bjorklund14} potentially paving the
way to break the $\OO^*(2^n)$ barrier for directed Hamiltonicity.

Turning to matchings, for counting maximum matching in general graphs we rely on~\cite{BjorklundRyser} as a black box for counting perfect matchings in general graphs. It would be very interesting if our technique could be used to simplify the algorithm of~\cite{BjorklundRyser}.

%% file: mttU.tex
\section{Matrix Tree Theorems}\label{app:a}
We present a proof of directed and undirected matrix tree theorem, that only rely on Cauchy-Binet Theorem. 

\mttU*
\begin{proof}
	From Cauchy-Binet Theorem, we can write,
	\[
		\det\left(A_rC^\top\right) = \sum_{\substack{S\subset E(G)\\|S|=|V(G)|-1}} \det(A_r[S])\det(C[S]).
	\]
	Hence, to prove the theorem, it would be enough to show that $\det(A_r[S])\det(C[S])$ is $\prod_{e\in S}x_e$ if $S\in \mathcal{T}(G)$, the set of all spanning trees of $G$, and zero otherwise.
	
	That is, we need to show that, for $S\subseteq E(G), |S|=|V(G)|-1$, 
	\begin{itemize}
		\item[(a)] if $S$ is not a spanning tree of $G$, then $\det (C[S])=0$,
		\item[(b)] if $S$ is a spanning tree of $G$, then  $\det(A_r[S])$ is $\sigma(S)\prod_{e\in S}x_e$, where $\sigma(S)\in\{-1,1\}$, and 
		\item[(c)] if $S$ is a spanning tree of $G$, then $\det(C[S])=\sigma(S)$.   
	\end{itemize} 
	
	\emph{Proof of (a)}. $S\notin\mathcal{T}(G)$, then $\det(C[S])$ is zero. Since $S$ is not a spanning tree, as per the premise, there must exist a connected component $Q$
	in $S$ that does not contain $r$. Then, summing together the rows of $C[S]$ corresponding to vertices in $Q$, we get an all zero row. This is because, for every column corresponding to an $e$ edge in $G[Q]$ has both its non-zero entries (of opposite sign) in the rows corresponding to the endpoints of $e$, both of which must be in $Q$. But this means that rows of $C[S]$ are not independent and hence $\det(C[S])$ must be zero.
	
	\emph{Proof of (b) and (c)}. 
	The idea is to show that the matrices $A_r[S]$ and $C[S]$ can be made lower triangular, by the same reordering of their rows and columns, such that the diagonal entries in the resulting matrices are all non-zero. Let the columns of $A_r[S]$ be indexed by $e_{i_1},\ldots,e_{i_{n-1}}$.  
	Let us assume that the tree $S$ is rooted at $r$. We describe the desired row and column permutations by an	
	iterative process that removes a leaf vertex in each round. We start with a leaf vertex $v$ (ties broken lexicographically) of $S$. Let $\{u,v\}$ be the edge in $S$ that is incident on $v$. We remove both the vertex $v$ and the edge $\{u,v\}$ from $S$, and call the resulting tree $S$. In general, in the $i$th round, we take a leaf vertex $w$ of $S$ and the edge $e$ incident on $w$, and remove them from $S$. We stop when $S$ is only the vertex $r$. Let $\pi_v=v_{i_1},v_{i_2},\ldots,v_{i_{n-1}}$ be the vertices of $S$ in the sequence they are removed from $S$, and let $\pi_e=e_{j_1},e_{j_2},\ldots,e_{j_{n-1}}$ be the edge of $S$ in the sequence they are removed. Clearly, $\pi_v$ and $\pi_e$ are permutations of the vertices and edges respectively, of $S$. We claim that, when the rows and columns of $A_r[S]$ are rearranged as in $\pi_v$ and $\pi_e$ respectively, the resulting matrix is a lower triangular matrix with all non-zero entries on the diagonal. To see this, we note that for any $k\in[n-1]$, the vertex $v_{i_k}$ has no edge $e_{j_\ell}$, for $k<\ell\leq n-1$ incident on it, since $v_{i_k}$ is a leaf node in the $k$th iteration. Also, the edge $e_{j_k}$ must be incident on $v_{i_k}$. Taken together, it means that, in the reordered matrix, $A'_r[S]$, the $k$th diagonal entry is non-zero and and all the entries in the row $k$ to the right of the diagonal are zero. Hence $A'_r[S]$ must be in the desired shape. Since, the matrices $A_r$ and $C$ agree on zero/non-zero entries including the sign, the same column and row permutations $\pi_v$ and $\pi_e$, must make $C[S]$ lower triangular with all non-zero entries on the diagonal. We know that the determinant of a lower triangular matrix is product of the diagonal entries, the determinant of $A'_r[S]$ should be $\tau(S)\prod_{e\in S} x_e$, where $\tau(S)$ is $1$ or $-1$ depending on the number of negative values on the diagonal. Similarly, the determinant of the reordered matrix $C[S]$, $C'[S]$, must be $\tau(S)$. In particular, working back to the original matrices we must have that, $\det(A_r[S])=\mbox{sign}(\pi_v)\det(A'_r[S])\mbox{sign}(\pi_e)=\sigma(S)\prod_{e\in S}x_e$ and $\det(C[S])=\mbox{sign}(\pi_e)\det(C'[S])\mbox{sign}(\pi_v)=\sigma(S)$.     
\end{proof}

Now, we see a proof of the directed matrix tree theorem, which gives a way to count the number of in-arborescences rooted at a particular vertex, via a determinant computation, similar to the undirected matrix tree theorem.
 

\mttD*
\begin{proof}
	The proof is similar to the proof of~\cref{lem:mttU}. 
	From Cauchy-Binet Theorem, we can write,
	\[
	\det\left(A_rC^\top\right) = \sum_{\substack{S\subset E(G)\\|S|=|V(G)|-1}} \det(A_r[S])\det(C[S]).
	\]
	From the proof of~\cref{lem:mttU} for part (a), we have that $\det(A_r[S])$ must be zero if $S$ is not a spanning tree of the undirected graph underlying $G$.
	Since $C$ is a $0$-$1$-valued matrix, with exactly one non-zero entry in each column, for $C[S]$ to be non-singular $C[S]$ must be a permutation matrix, i.e., each row and column of $C[S]$ has exactly one non-zero entry. This ensures that in the graph induced by $S$, the all the vertices, except for $r$ (row $r$ is not present in $C$), must have in-degree exactly one, by definition of $C$. Hence, we are certain that $\det(A_r[S])\det(C[S])$ is non-zero if and only if $S$ is an arborescence rooted at $r$. In fact, we can show that $\det(A_r[S])\det(C[S])$ is $\prod_{e\in S}x_e$ if $S$ is an $r$-rooted in-arborescence of $G$. We can use the same reordering of rows and columns of $A_r[S]$ and $C[S]$, as in the proof of~\cref{lem:mttU} for part (b) and (c) to bring convert them into lower triangular matrices. The orientation as per the arborescence, ensures that the signs of both the determinant are same. 
\end{proof} 

%% file: biblio.bib
@article{HAB,
	title = {Uniform constant-depth threshold circuits for division and iterated multiplication},
	journal = {Journal of Computer and System Sciences},
	volume = {65},
	number = {4},
	pages = {695-716},
	year = {2002},
	note = {Special Issue on Complexity 2001},
	issn = {0022-0000},
	doi = {https://doi.org/10.1016/S0022-0000(02)00025-9},
	url = {https://www.sciencedirect.com/science/article/pii/S0022000002000259},
	author = {William Hesse and Eric Allender and David A. {Barrington}},
}

@article{Valiant,
  author       = {Leslie G. Valiant},
  title        = {The Complexity of Computing the Permanent},
  journal      = {Theor. Comput. Sci.},
  volume       = {8},
  pages        = {189--201},
  year         = {1979},
  url          = {https://doi.org/10.1016/0304-3975(79)90044-6},
  doi          = {10.1016/0304-3975(79)90044-6},
  bibsource    = {dblp computer science bibliography, https://dblp.org}
}

@book{Stanley,
  title={Algebraic Combinatorics: Walks, Trees, Tableaux, and More},
  author={Richard P. Stanley},
  year={2018},
  publisher={Springer-Verlag},
  series={Undergraduate Texts in Mathematics},
}

@article{Berk,
  author       = {Stuart J. Berkowitz},
  title        = {On Computing the Determinant in Small Parallel Time Using a Small
                  Number of Processors},
  journal      = {Inf. Process. Lett.},
  volume       = {18},
  number       = {3},
  pages        = {147--150},
  year         = {1984},
  url          = {https://doi.org/10.1016/0020-0190(84)90018-8},
  doi          = {10.1016/0020-0190(84)90018-8},
  bibsource    = {dblp computer science bibliography, https://dblp.org}
}

@article{MV,
  author       = {Meena Mahajan and
                  V. Vinay},
  title        = {Determinant: Combinatorics, Algorithms, and Complexity},
  journal      = {Chic. J. Theor. Comput. Sci.},
  volume       = {1997},
  year         = {1997},
  url          = {http://cjtcs.cs.uchicago.edu/articles/1997/5/contents.html},
  bibsource    = {dblp computer science bibliography, https://dblp.org}
}

@book{vzGG, place={Cambridge}, edition={3}, title={Modern Computer Algebra}, publisher={Cambridge University Press}, author={von zur Gathen, Joachim and Gerhard, Jürgen}, year={2013}}

@article{kuhn,
  title={The Hungarian method for the assignment problem},
  author={Kuhn, Harold W},
  journal={Naval Research Logistics Quarterly},
  volume={2},
  number={1-2},
  pages={83--97},
  year={1955},
  publisher={Wiley Online Library}
}

@inproceedings{FGS,
  author       = {Michael A. Forbes and
                  Sumanta Ghosh and
                  Nitin Saxena},
  editor       = {Ioannis Chatzigiannakis and
                  Christos Kaklamanis and
                  D{\'{a}}niel Marx and
                  Donald Sannella},
  title        = {Towards Blackbox Identity Testing of Log-Variate Circuits},
  booktitle    = {45th International Colloquium on Automata, Languages, and Programming,
                  {ICALP} 2018, July 9-13, 2018, Prague, Czech Republic},
  series       = {LIPIcs},
  volume       = {107},
  pages        = {54:1--54:16},
  publisher    = {Schloss Dagstuhl - Leibniz-Zentrum f{\"{u}}r Informatik},
  year         = {2018},
  url          = {https://doi.org/10.4230/LIPIcs.ICALP.2018.54},
  doi          = {10.4230/LIPICS.ICALP.2018.54},
  bibsource    = {dblp computer science bibliography, https://dblp.org}
}

@article{HeldKarp,
  author       = {Michael Held and
                  Richard M. Karp},
  title        = {The Traveling-Salesman Problem and Minimum Spanning Trees},
  journal      = {Oper. Res.},
  volume       = {18},
  number       = {6},
  pages        = {1138--1162},
  year         = {1970},
  url          = {https://doi.org/10.1287/opre.18.6.1138},
  doi          = {10.1287/OPRE.18.6.1138},
  bibsource    = {dblp computer science bibliography, https://dblp.org}
}

@article{MVV,
  author       = {Ketan Mulmuley and
                  Umesh V. Vazirani and
                  Vijay V. Vazirani},
  title        = {Matching is as easy as matrix inversion},
  journal      = {Comb.},
  volume       = {7},
  number       = {1},
  pages        = {105--113},
  year         = {1987},
  url          = {https://doi.org/10.1007/BF02579206},
  doi          = {10.1007/BF02579206},
  bibsource    = {dblp computer science bibliography, https://dblp.org}
}

@article{Bjorklund14,
  author       = {Andreas Bj{\"{o}}rklund},
  title        = {Determinant Sums for Undirected Hamiltonicity},
  journal      = {{SIAM} J. Comput.},
  volume       = {43},
  number       = {1},
  pages        = {280--299},
  year         = {2014},
  url          = {https://doi.org/10.1137/110839229},
  doi          = {10.1137/110839229},
  bibsource    = {dblp computer science bibliography, https://dblp.org}
}

@article{Barvinok,
  author       = {Alexander I. Barvinok},
  title        = {Two Algorithmic Results for the Traveling Salesman Problem},
  journal      = {Math. Oper. Res.},
  volume       = {21},
  number       = {1},
  pages        = {65--84},
  year         = {1996},
  url          = {https://doi.org/10.1287/moor.21.1.65},
  doi          = {10.1287/MOOR.21.1.65},
  bibsource    = {dblp computer science bibliography, https://dblp.org}
}

@inproceedings{BjorklundRyser,
  author       = {Andreas Bj{\"{o}}rklund},
  editor       = {Yuval Rabani},
  title        = {Counting perfect matchings as fast as Ryser},
  booktitle    = {Proceedings of the Twenty-Third Annual {ACM-SIAM} Symposium on Discrete
                  Algorithms, {SODA} 2012, Kyoto, Japan, January 17-19, 2012},
  pages        = {914--921},
  publisher    = {{SIAM}},
  year         = {2012},
  url          = {https://doi.org/10.1137/1.9781611973099.73},
  doi          = {10.1137/1.9781611973099.73},
  bibsource    = {dblp computer science bibliography, https://dblp.org}
}

@book{Ryser,
  title={Combinatorial Mathematics},
  author={Ryser, H.J.},
  isbn={9781614440147},
  series={The Carus Mathematical Monographs},
  url={https://books.google.co.in/books?id=e-GyDwAAQBAJ},
  year={1963},
  publisher={Mathematical Association of America}
}

@article{KoutisWilliams,
  author       = {Ioannis Koutis and
                  Ryan Williams},
  title        = {Algebraic fingerprints for faster algorithms},
  journal      = {Commun. {ACM}},
  volume       = {59},
  number       = {1},
  pages        = {98--105},
  year         = {2016},
  url          = {https://doi.org/10.1145/2742544},
  doi          = {10.1145/2742544},
  bibsource    = {dblp computer science bibliography, https://dblp.org}
}

@inproceedings{BKK17,
  author       = {Andreas Bj{\"{o}}rklund and
                  Petteri Kaski and
                  Ioannis Koutis},
  title        = {Directed Hamiltonicity and Out-Branchings via Generalized Laplacians},
  booktitle    = {44th International Colloquium on Automata, Languages, and Programming,
                  {ICALP} 2017, July 10-14, 2017, Warsaw, Poland},
  pages        = {91:1--91:14},
  year         = {2017},
  url          = {https://doi.org/10.4230/LIPIcs.ICALP.2017.91},
  doi          = {10.4230/LIPICS.ICALP.2017.91},
  bibsource    = {dblp computer science bibliography, https://dblp.org}
}

@inproceedings{EKW24,
  author       = {Eduard Eiben and
                  Tomohiro Koana and
                  Magnus Wahlstr{\"{o}}m},
  editor       = {David P. Woodruff},
  title        = {Determinantal Sieving},
  booktitle    = {Proceedings of the 2024 {ACM-SIAM} Symposium on Discrete Algorithms,
                  {SODA} 2024, Alexandria, VA, USA, January 7-10, 2024},
  pages        = {377--423},
  publisher    = {{SIAM}},
  year         = {2024},
  url          = {https://doi.org/10.1137/1.9781611977912.16},
  doi          = {10.1137/1.9781611977912.16},
  bibsource    = {dblp computer science bibliography, https://dblp.org}
}

@article{AAM,
  author       = {Eric Allender and
                  Vikraman Arvind and
                  Meena Mahajan},
  title        = {Arithmetic Complexity, Kleene Closure, and Formal Power Series},
  journal      = {Theory Comput. Syst.},
  volume       = {36},
  number       = {4},
  pages        = {303--328},
  year         = {2003},
  url          = {https://doi.org/10.1007/s00224-003-1077-7},
  doi          = {10.1007/S00224-003-1077-7},
  bibsource    = {dblp computer science bibliography, https://dblp.org}
}

@book{LP,
  title={Matching Theory},
  author={Lov{\'a}sz, L{\'a}szl{\'o} and Plummer, Michael D.},
  year={1986},
  publisher={North-Holland},
  series={Annals of Discrete Mathematics},
  volume={29}
}

@article{Edmonds, title={Paths, Trees, and Flowers}, volume={17}, DOI={10.4153/CJM-1965-045-4}, journal={Canadian Journal of Mathematics}, author={Edmonds, Jack}, year={1965}, pages={449–467}}

@inproceedings{Curt,
  author       = {Radu Curticapean},
  editor       = {Fedor V. Fomin and
                  Rusins Freivalds and
                  Marta Z. Kwiatkowska and
                  David Peleg},
  title        = {Counting Matchings of Size k Is W[1]-Hard},
  booktitle    = {Automata, Languages, and Programming - 40th International Colloquium,
                  {ICALP} 2013, Riga, Latvia, July 8-12, 2013, Proceedings, Part {I}},
  series       = {Lecture Notes in Computer Science},
  volume       = {7965},
  pages        = {352--363},
  publisher    = {Springer},
  year         = {2013},
  url          = {https://doi.org/10.1007/978-3-642-39206-1\_30},
  doi          = {10.1007/978-3-642-39206-1\_30},
  bibsource    = {dblp computer science bibliography, https://dblp.org}
}

@article{ACDM18,
  author       = {Vikraman Arvind and
                  Abhranil Chatterjee and
                  Rajit Datta and
                  Partha Mukhopadhyay},
  title        = {Fast Exact Algorithms Using Hadamard Product of Polynomials},
  journal      = {Algorithmica},
  volume       = {84},
  number       = {2},
  pages        = {436--463},
  year         = {2022},
  url          = {https://doi.org/10.1007/s00453-021-00900-0},
  doi          = {10.1007/S00453-021-00900-0},
  bibsource    = {dblp computer science bibliography, https://dblp.org}
}

@book{GJ,
  author       = {M. R. Garey and
                  David S. Johnson},
  title        = {Computers and Intractability: {A} Guide to the Theory of NP-Completeness},
  publisher    = {W. H. Freeman},
  year         = {1979},
  isbn         = {0-7167-1044-7},
  bibsource    = {dblp computer science bibliography, https://dblp.org}
}

@inproceedings{BDH18,
	author       = {Cornelius Brand and
	Holger Dell and
	Thore Husfeldt},
	editor       = {Ilias Diakonikolas and
	David Kempe and
	Monika Henzinger},
	title        = {Extensor-coding},
	booktitle    = {Proceedings of the 50th Annual {ACM} {SIGACT} Symposium on Theory
	of Computing, {STOC} 2018, Los Angeles, CA, USA, June 25-29, 2018},
	pages        = {151--164},
	publisher    = {{ACM}},
	year         = {2018},
	url          = {https://doi.org/10.1145/3188745.3188902},
	doi          = {10.1145/3188745.3188902},
	bibsource    = {dblp computer science bibliography, https://dblp.org}
}

@inproceedings{Koutis08,
	author       = {Ioannis Koutis},
	editor       = {Luca Aceto and
	Ivan Damg{\aa}rd and
	Leslie Ann Goldberg and
	Magn{\'{u}}s M. Halld{\'{o}}rsson and
	Anna Ing{\'{o}}lfsd{\'{o}}ttir and
	Igor Walukiewicz},
	title        = {Faster Algebraic Algorithms for Path and Packing Problems},
	booktitle    = {Automata, Languages and Programming, 35th International Colloquium,
	{ICALP} 2008, Reykjavik, Iceland, July 7-11, 2008, Proceedings, Part
	{I:} Tack {A:} Algorithms, Automata, Complexity, and Games},
	series       = {Lecture Notes in Computer Science},
	volume       = {5125},
	pages        = {575--586},
	publisher    = {Springer},
	year         = {2008},
	url          = {https://doi.org/10.1007/978-3-540-70575-8\_47},
	doi          = {10.1007/978-3-540-70575-8\_47},
	bibsource    = {dblp computer science bibliography, https://dblp.org}
}

@article{Williams09,
	author       = {Ryan Williams},
	title        = {Finding paths of length k in O\({}^{\mbox{*}}\)(2\({}^{\mbox{k}}\))
	time},
	journal      = {Inf. Process. Lett.},
	volume       = {109},
	number       = {6},
	pages        = {315--318},
	year         = {2009},
	url          = {https://doi.org/10.1016/j.ipl.2008.11.004},
	doi          = {10.1016/J.IPL.2008.11.004},
	bibsource    = {dblp computer science bibliography, https://dblp.org}
}

@article{BCKN15,
	title = {Deterministic single exponential time algorithms for connectivity problems parameterized by treewidth},
	journal = {Information and Computation},
	volume = {243},
	pages = {86-111},
	year = {2015},
	note = {40th International Colloquium on Automata, Languages and Programming (ICALP 2013)},
	issn = {0890-5401},
	doi = {https://doi.org/10.1016/j.ic.2014.12.008},
	url = {https://www.sciencedirect.com/science/article/pii/S0890540114001606},
	author = {Hans L. Bodlaender and Marek Cygan and Stefan Kratsch and Jesper Nederlof}
}

@InProceedings{Wlodarczyk16,
	author =	{Wlodarczyk, Michal},
	title =	{{Clifford Algebras Meet Tree Decompositions}},
	booktitle =	{11th International Symposium on Parameterized and Exact Computation (IPEC 2016)},
	pages =	{29:1--29:18},
	series =	{Leibniz International Proceedings in Informatics (LIPIcs)},
	ISBN =	{978-3-95977-023-1},
	ISSN =	{1868-8969},
	year =	{2017},
	volume =	{63},
	editor =	{Guo, Jiong and Hermelin, Danny},
	publisher =	{Schloss Dagstuhl -- Leibniz-Zentrum f{\"u}r Informatik},
	address =	{Dagstuhl, Germany},
	URL =		{https://drops.dagstuhl.de/entities/document/10.4230/LIPIcs.IPEC.2016.29},
	URN =		{urn:nbn:de:0030-drops-69260},
	doi =		{10.4230/LIPIcs.IPEC.2016.29}
}

@inproceedings{Pratt19,
	author       = {Kevin Pratt},
	editor       = {David Zuckerman},
	title        = {Waring Rank, Parameterized and Exact Algorithms},
	booktitle    = {60th {IEEE} Annual Symposium on Foundations of Computer Science, {FOCS}
	2019, Baltimore, Maryland, USA, November 9-12, 2019},
	pages        = {806--823},
	publisher    = {{IEEE} Computer Society},
	year         = {2019},
	url          = {https://doi.org/10.1109/FOCS.2019.00053},
	doi          = {10.1109/FOCS.2019.00053},
	bibsource    = {dblp computer science bibliography, https://dblp.org}
}

@inproceedings{BKS23,
	author       = {Cornelius Brand and
	Viktoriia Korchemna and
	Michael Skotnica},
	editor       = {J{\'{e}}r{\^{o}}me Leroux and
	Sylvain Lombardy and
	David Peleg},
	title        = {Deterministic Constrained Multilinear Detection},
	booktitle    = {48th International Symposium on Mathematical Foundations of Computer
	Science, {MFCS} 2023, Bordeaux, France, August 28 - September 1, 2023},
	series       = {LIPIcs},
	volume       = {272},
	pages        = {25:1--25:14},
	publisher    = {Schloss Dagstuhl - Leibniz-Zentrum f{\"{u}}r Informatik},
	year         = {2023},
	url          = {https://doi.org/10.4230/LIPIcs.MFCS.2023.25},
	doi          = {10.4230/LIPICS.MFCS.2023.25},
	bibsource    = {dblp computer science bibliography, https://dblp.org}
}

@article{Brand22,
	author       = {Cornelius Brand},
	title        = {A note on algebraic techniques for subgraph detection},
	journal      = {Inf. Process. Lett.},
	volume       = {176},
	pages        = {106242},
	year         = {2022},
	url          = {https://doi.org/10.1016/j.ipl.2021.106242},
	doi          = {10.1016/J.IPL.2021.106242},
	bibsource    = {dblp computer science bibliography, https://dblp.org}
}

@inproceedings{Bjorklund11,
	title={Counting perfect matchings as fast as Ryser},
	author={Andreas Bj{\"o}rklund},
	booktitle={ACM-SIAM Symposium on Discrete Algorithms},
	year={2011},
	url={https://api.semanticscholar.org/CorpusID:6919952}
}
